\documentclass[11pt,a4paper]{article}
\usepackage{epsf,epsfig,amsfonts,amsgen,amsmath,amstext,amsbsy,amsopn,amsthm}
\usepackage{amssymb}
\usepackage{cite}
\usepackage{enumitem}
\newtheorem{theorem}{Theorem}

\newtheorem{prop}{Proposition}
\newtheorem{lemma}[theorem]{Lemma}

\theoremstyle{definition}

\newtheorem{conj}[theorem]{Conjecture}

\newtheorem{example}{Example}
\begin{document}
\title{\bf {\Large Dimensions of nonbinary antiprimitive BCH codes and some conjectures}}
\date{}
\author{Yang Liu, Ruihu Li$^{\dag}$, Luobin Guo, Hao Song\\
 Department of Basic Sciences,  Air Force
Engineering University, Xi'an,  \\Shaanxi 710051, China.
(email:$^{\dag}$llzsy2015@163.com, liu\_yang10@163.com)\\
} \maketitle
\begin{abstract}
Bose-Chaudhuri-Hocquenghem (BCH) codes have been intensively
investigated. Even so, there is only a little  known  about
primitive BCH codes, let alone non-primitive ones. In this paper,
let $q>2$ be a prime power, the dimension of a family of
non-primitive  BCH codes of length $n=q^{m}+1$ (also called
antiprimitive) is studied. These codes are also linear codes with
complementary duals (called LCD codes). Through some approaches such
as iterative algorithm, partition  and scaling,  all coset leaders
of $C_{x}$ modulo $n$ with $q^{\lceil \frac{m}{2}\rceil}<x\leq
2q^{\lceil\frac{m}{2} \rceil}+2$ are given for $m\geq 4$. And for
odd $m$ the first several largest coset leaders modulo $n$ are
determined. Furthermore, a new  kind of sequences is introduced to
determine the second largest coset leader modulo $n$ with $m$ even
and $q$ odd. Also, for even $m$ some conjectures  about the first
several coset leaders modulo $n$ are proposed,  whose complete
verification would wipe out the difficult problem to determine the
first several coset leaders of antiprimitive  BCH codes.  After
deriving the cardinalities of the coset leaders, we shall calculate
exact dimensions of many antiprimitive LCD BCH codes.
%It is worth  mentioning  that their   maximum designed distance come
%up  to that of the narrow-sense ones.
\medskip

\end{abstract}

\section{\label{sec:level1} Introduction\protect}

\label{sec1} Binary BCH codes were introduced by Hocquenghem in 1959
\cite{Hocquenghem}, and independently by Bose and Ray-Chaudhuri in
1960 \cite{Bose1,Bose2}. Immediately afterwards they were
 generalized by Gorenstein and Zierler  to general finite fields \cite{Gor}.
Since then, BCH codes have been  deeply studied and widely developed
because of their advantages of good error-correcting capability and
efficient  encoding and decoding algorithms.

However, as pointed out by  Charpin \cite{Char} and Ding
\cite{Ding5}, it is a hard problem to determine the dimension and
minimum distance of BCH codes. Therefore, their parameters are
determined for only some special classes of code lengths.

Actually, the research of the dimension of BCH codes began at the
nearly same time when they were discovered \cite{Mann}. For
narrow-sense primitive BCH codes, many important conclusions
 on the dimension have been obtained
\cite{Yue,Aly1,Aly2,Ding1,Ding2,Ding3,Ding4,Ding6,Ding7,Ding8},
while only some sparse results for the nonprimitive case
\cite{Ding7,Ding8,Ding9,Shixin,liu }. Thereinto, Liu {\it et al.}
\cite{Ding7} and Li {\it et al.} \cite{Ding8} successively studied a
kind of  BCH codes of length $n=q^m+1$.  They acquired some
important achievements about the parameters of BCH codes of length
$n$ for designed distance $\delta \leq q^{\lfloor \frac{m-1}{2}
\rfloor}+3$ in \cite{Ding7} and
 $q^{\lfloor\frac{m-1}{2}\rfloor}+3<\delta\leq q^{\lceil
\frac{m}{2} \rceil}$ in \cite{Ding8}, respectively. Also, they
proved that this  kind of  BCH codes is LCD codes. LCD codes were
initially named as reversible codes in \cite{Massey1} and introduced
in \cite{Massey2} by Massey. Recently, it was   found that they can
improve information security, especially against side-channel
attacks and fault noninvasive attacks in cryptography \cite{Carlet}.
Ding pointed out that it is important but not easy to find the
second and third largest coset leaders $\delta_{2}$ and $\delta_{3}$
modulo $n=q^m+1$ \cite{Ding5}, which is quite helpful to obtain both
the dimension of  BCH codes  and the Bose distance (i.e., the
maximal designed distance). In \cite{liu },  we introduced  some new
techniques to find out the coset leaders modulo $n=2^m+1$ with
$m\not \equiv 0 \bmod 8$. And we  determined the first five largest
coset leaders modulo $n$, as well as all coset leaders for
$$x\leq \left\{
\begin{array}{lll}
2^{t+2}+7
   &\mbox {if $m=2t+1$};\\
 2^{2t+2}+2^{2t+1}+3 &\mbox {if $m=4t+2$};\\
 2^{4t+3}+2^{4t+2}+2^{4t+1}+1 &\mbox {if $m=8t+4$.}
 \end{array}
 \right.$$
 Motivated by the   cryptographic importance of LCD codes \cite{Carlet} and the
 valuable report of Ding in \cite{Ding5}, this article  aims to  study the dimension
of nonbinary LCD BCH codes of length $n=q^m+1$ based on Refs.
\cite{Ding7,Ding8,liu }. It extends our previous work on binary
antiprimitive LCD  BCH codes \cite{liu }, simplifies some of the
proofs and generalizes many of the results to the nonbinary case.
The main results  are listed as follows.

(1): For   $m\geq 4$, the coset leaders of $C_{x}$ modulo
$n=q^{m}+1$ are determined for
$$q^{\lceil\frac{m}{2}\rceil}<x\leq\left\{
\begin{array}{lll}
 2q^{\lceil\frac{m}{2}\rceil}+2  &\mbox {if $m=2t$;}\\
2q^{\lceil\frac{m}{2}\rceil}+2q-1 &\mbox {if $m=2t+1$.}
 \end{array}
 \right.$$
This doubly extends  the corresponding results of Refs.
\cite{Ding7,Ding8} (in detail, $x<q^{\lceil\frac{m}{2}\rceil}$)
about the cosets leaders modulo $n$.

(2): A new  kind of  useful sequences is introduced. Consequently,
for even $m$ and odd $q$  the second largest coset leader modulo
$n=q^{m}+1$ is determined. Also, when $m$ is odd, the first six
(resp. five) largest coset leaders modulo $n=q^{m}+1$ are determined
with $q$ odd (resp. even). These results above  shall solve a
majority of the significative problem proposed by Ding \cite{Ding5}
to determine the second and third largest coset leaders.

Furthermore, some  conjectures about the first several largest coset
leaders modulo $n$ are proposed for  even $m$. If they could be well
verified, the problem about first several largest coset leaders of
antiprimitive  BCH codes would be completely settled.

(3): Through calculating the cardinalities of relevant cyclotomic
cosets,  dimensions of some nonbinary antiprimitive   BCH codes are
precisely obtained. Additionally, their Bose distances are given
together.

This article is organized as follows. In Section 2, some basic
concepts on cyclotomic cosets,  BCH codes and LCD codes are
reviewed. In Section 3, the parameters of BCH codes of length
$n=q^{m}+1 (m\geq 4)$ with designed distances   for $q^{\lceil
\frac{m}{2}\rceil}<\delta \leq 2q^{\lceil \frac{m}{2}\rceil}+2$. In
Section 4, the parameters of BCH codes of length $n=q^{2t+1}+1$ with
designed distances $\delta>\delta_{5}$ (resp. $\delta>\delta_{6}$)
are determined when $q$ is even (resp. odd). Some conclusions and
conjectures about the first several coset leaders modulo $n=q^{m}+1$
with $m$ even are presented in Section 5.
 The final remarks are drawn in Section 6.

\section{Preliminaries}
\label{sec2}

In this section, we recall some basic concepts on cyclotomic cosets,
cyclic codes, BCH codes and LCD codes. For more details, one can
refer to Refs. \cite{Carlet,mac,huf}.

Let $q$ be a prime power and $\mathbb{F}_{q}$ be the finite field
with $q$ elements.   A linear $[n,k,d]_q$ code $\mathcal{C}$  is
denoted as  a $k$-dimensional subspace of $\mathbb{F}^{n}_{q}$ with
minimum (Hamming) distance $d$.
 $\mathcal{C}$ is  cyclic if
  $(c_{0},c_{1},\cdots,c_{n-1}) \in \mathcal{C}$  implies
  $(c_{n-1},c_{0},c_{1},\cdots,c_{n-2})$ $\in$ $\mathcal{C}$.
By associating each vector $ (c_{0},c_{1},\cdots,c_{n-1})$
 $\in \mathbb{F}^{n}_{q}$
  with a polynomial $c(x)=c_{0}+c_{1}x+\cdots+c_{n-1}x^{n-1}$
  $\in \mathbb{F}_{q}[x]/(x^{n}-1),$
 then every cyclic code $\mathcal{C}$ is identified with an ideal of
  $\mathbb{F}_{q}[x]/(x^{n}-1)$.
Since every ideal of $ \mathbb{F}_{q}[x]/(x^{n}-1)$ is  principal,
each cyclic code $\mathcal{C}$ can be identified with
$\mathcal{C}=\langle g(x)\rangle$, where $g(x)$ is monic  and has
the smallest degree among all the generators of $\mathcal{C}$. This
polynomial  g(x)  is called the {\it generator polynomial} of
$\mathcal{C}$.

Denote $\mathbb{Z}_n=\{0,1,2,\cdots,n-1\}$. If $\gcd(q,n)=1$ and
$x\in \mathbb{Z}_n$,  a {\it $q$-cyclotomic coset} modulo $n$
containing $x$ is defined by $$C_{x}=\{xq^{i}\bmod n| 0\leq i \leq
l-1\}\subseteq \mathbb{Z}_n,$$ where $l$ is the smallest positive
integer such that $q^{l}x\equiv x \bmod n$. The cardinality of
$C_{x}$ is denoted by $|C_{x}|=l$. The smallest integer in $C_{x}$
is called the {\it coset leader} of $C_{x}$ modulo $n$.

If $\xi$ is a primitive $n$-th root of unity in some field
containing $\mathbb{F}_{q}$, $T$$=\{i|g(\xi^{i})=0\}$ is called the
{\it defining set} of $\mathcal{C}=\langle g(x)\rangle$. It is well
known that $T$ is the union of some $q$-cyclotomic cosets modulo
$n$.  The dimension $k$ of $\mathcal{C}$ is determined by $k=n-|T|$
and the minimum distance $d$ can be evaluated by $T$.

 $\mathcal{C}$ is called a {\it BCH code} of designed distance
$\delta$ if $T=C_{b}\cup C_{b+1}\cup \cdots \cup C_{b+\delta-2}$.
And $\mathcal{C}$ can be  denoted by $\mathcal{C}(n,q,\delta,b)$ as
given in \cite{Ding7,Ding8}. If $b=1$, $\mathcal{C}$ is called a
narrow-sense BCH code, and non-narrow-sense, otherwise. If
$n=q^{m}-1$, $\mathcal{C}$ is called {\it primitive},
 and {\it non-primitive}, otherwise.
Particularly, if $n=q^{m}+1$,  it is called  {\it antiprimitive}
  by  Ding in \cite{Ding5}.

Given two vectors $\mathbf{x}=(x_1,x_2,\cdots , x_n)$ and
$\mathbf{y}=(y_1, y_2,\cdots, y_n)$$\in \mathbb{F}_{q}^{n}$, their
Euclidean inner product is denoted by $(\mathbf{x},
\mathbf{y})=x_1y_1+x_2y_2+\cdots+x_ny_n.$

The Euclidean dual code $\mathcal{C}^{\perp }$ of $\mathcal{C}$ is
defined by $\mathcal{C}^{\perp }=\{ \mathbf{x} \in \mathbb{F}_{q}^n
\mid (\mathbf{x}, \mathbf{y})=0, \forall~ \mathbf{y} \in \mathcal{C}
\}.$ $\mathcal{C}$ is called an {\it LCD code} if
$\mathcal{C}^{\perp }\bigcap\mathcal{C}=\{\bf 0\}$, which is
equivalent to $\mathcal{C}^{\perp
}\bigoplus\mathcal{C}=\mathbb{F}_q^n$.

Below,  we  pay main attention to nonbinary antiprimitive    BCH
codes with defining sets $T=\bigcup\limits_{i=1}^{\delta-1}C_{i}$
 and $T_{0}=\{0\}\bigcup T$, which can be denoted by
$\mathcal{C}(n,q,\delta,1)$ and $\mathcal{C}(n,q,\delta+1,0)$,
respectively. As thus, they have parameters
$\mathcal{C}(n,q,\delta,1)=[n, n-|T|,d\geq \delta]_{q}$ and
$\mathcal{C}(n,q,\delta+1,0)=[n, q^m-|T|,d_{0}\geq 2\delta]_{q}$,
respectively.

Throughout this paper,  let $q$ be a prime power and $n=q^m+1$.
 Suppose that $a,b,c,x \in \mathbb{Z}_n$
and $a\leq b$.    We denote $\{x| a\leq x \leq b\}$ by $[a,b]$ and
define $[a,b]+c=[a+c,b+c]$. Let ``$x \equiv y$" denote ``$x\equiv y
\bmod n$" for short unless otherwise noted. The words ``modulo $n$"
are omitted when cyclotomic cosets and coset leaders  are mentioned.
For example, ``$x$ is a coset leader" means ``$x$ is a coset leader
of $C_x$ modulo $n$". It is obvious that if $x\neq 0$, $x\in
\mathbb{Z}_n$ and $q|x$, then $\frac{x}{q}\in C_x$ and $x$ is not a
coset leader. Therefore, to determine coset leaders in $T$, one only
needs to consider $x$ with $q\nmid x$.

 Notation 3 in \cite{liu} shall be naturally generalized to the
 following  result.

\begin{prop}\label{prop}
{Let $n=q^m+1$ and $C_{x}$ be a cyclotomic coset modulo $n$ containing $x\in
 \mathbb{Z}_n$. For $0\leq k\leq m-1$, define  $$y_{_{x,k}}\equiv q^k x \bmod n \hbox{~with~} y_{_{x,k}}\in
\mathbb{Z}_n.$$
 Then
 (1)  $C_{x}$ can be denoted by $C_{x}=\{y_{_{x,k}},n-y_{_{x,k}}| 0\leq k\leq
 m-1\};$

 (2) $x$ is the coset leader of $C_{x}$  if and only if $y_{_{x,k}}-x \geq 0$ and
$n-y_{_{x,k}}-x \geq 0$ for $0\leq k \leq m-1$}.
\end{prop}

 The proposition above is quite necessary to
this article and will be continually utilized below. For clarity, an
example is provided as follows.

\begin{example}{If $n=3^2+1=10$, we shall obtain the following cyclotomic
cosets by Proposition \ref{prop}.

  $C_{1}=\{y_{_{1,0}}=1,y_{_{1,1}}=3,y_{_{1,2}}=10-1=9,y_{_{1,3}}=10-3=7\}=\{1,3,9,7\}=C_{7}$,

  $C_{2}=\{y_{_{2,0}}=2,y_{_{2,1}}=6,y_{_{2,2}}=10-2=8,y_{_{2,3}}=10-6=4\}=\{2,6,8,4\}=C_{4}=C_{8}$,

$C_{5}=\{y_{_{5,0}}=5,y_{_{5,1}}=5,y_{_{5,2}}=10-5=5,y_{_{5,3}}=10-5=5\}=\{5\}$,

It is easy to know that 1, 2 and 5 are coset leaders, while 4, 7 and
8 are not. Since 3 and 6 are both divisible by 3, they are out of
consideration.}\end{example}

\section{Dimensions of BCH codes with relatively small distance}

In this section, suppose that $n=q^{m}+1$ with $m\geq 4$. The coset
leaders of $C_{x}$ are determined for $$x\leq\left\{
\begin{array}{lll}
 2q^{t}+2  &\mbox {if $m=2t$;}\\
2q^{t+1}+2q-1 &\mbox {if $m=2t+1$.}
 \end{array}
 \right.$$  We split
into two subsections according to the parity  of $m$. In each
subsection, the coset leaders and the cardinalities of  the
cyclotomic cosets containing them are firstly
 determined. Then the dimension of many BCH codes with relatively
 small designed distance shall be naturally calculated.
\subsection{ BCH codes of length $n=q^{m}+1$ for $m$ odd  }
Throughout this subsection, let  $n=q^{m}+1$ with $m=2t+1\geq 5$.
\begin{theorem}\label{ther3.1} If  $x \not \equiv 0 \bmod q$,
 then the following statements hold:

(1) If $1\leq x\leq q^{t+1}-q-1$, then $x$ is a coset leader (see
\cite{Ding7,Ding8}).

(2) If $ q^{t+1}+ q+1 \leq x \leq q^{t+1}+q^{t}-2$,  then $x$ is a
coset leader.

(3) Given an integer $\alpha\in[1, q-2]$. If $ q^{t+1}+\alpha
q^{t}+2\leq x \leq   q^{t+1}+(\alpha +1) q^{t}-2$,
 then $x$ is a coset leader.

(4) If $q^{t+1}+(q-1) q^{t}+2 \leq x\leq 2q^{t+1}-2q-1$, then $x$ is
a coset leader.

(5)  Given  three integers $\alpha\in[1, q-1]$, $\beta=1~\hbox{or}~
2$ and  $1\leq \gamma \leq \beta q-1$. If
  $$x=q^{t+1}+\alpha  q^{t}\pm 1 ~~ \hbox{or}~~ x = \beta q^{t+1}\pm \gamma,$$
 then $x$ is not a coset leader.
\end{theorem}

\begin{proof}
 See Appendix \ref{pther3.1}.
\end{proof}

\begin{lemma}\label{lemm3.2}
 If $1 \leq x \leq  2q^{t+1}-2q-1$,  then $|C_x|=2m$.
\end{lemma}

\begin{proof} Seeking a contradiction,
suppose $|C_x|=k$  with $k<2m$.  It follows that  $x(q^{k}-1) \equiv
0$. Since $k|2m$ and $m$ is odd,    we have  $k=m, \frac{2m}{3}$ or
$k \leq \frac{2m}{5}$.

(1) If $k=m$, then $(q^{k}-1, n)=(q^{m}-1, q^{m}+1)= 2$. From $1
\leq x \leq q^{t+1}-2q-1<\frac{n}{2}$, one has $x(q^{k}-1) \not
\equiv 0$,  a contradiction.

(2)  If $m\equiv 0\bmod  3$ and $k=\frac{2m}{3}$, then $(q^{k}-1,
n=q^{m}+1)=q^{\frac{m}{3}}+1$. Since $1 \leq x \leq
q^{t+1}-2q-1<\frac{n}{q^{\frac{m}{3}}+1}$
$=q^{\frac{2m}{3}}-q^{\frac{m}{3}}+1$, so there holds $x(q^{k}-1)
\not \equiv 0 $.

(3)  If $k\leq \frac{2m}{5}$, then $1\leq x(q^{k}-1)\leq
(2q^{t+1}-2q-1)(q^{\frac{2m}{5}}-1) <n$ and $x(q^{k}-1) \not \equiv
0$.

Collecting all previous discussions, we have $x(q^{k}-1) \not \equiv
0 $, this yields a contradiction. Hence,  $|C_x|=2m$ for $1 \leq x
\leq 2q^{t+1}-2q-1$.
\end{proof}

The previous results on the coset leaders and  cardinalities are
sufficient to give the following conclusion.

\begin{theorem}\label{theo3.7} Suppose that $n=q^m+1$ with $m=2t+1 \geq 5$. If $\alpha \in[1,q-2]$,
 then the following statements hold:

 (1) The narrow-sense BCH codes $\mathcal{C}(n,q,\delta,1)$ have
parameters
\begin{scriptsize}
 $$ \left\{
\begin{array}{lll}
\hbox{[}n, n-2m\lceil(\delta-1)(1-1/q)\rceil+4m(q-1), d \geq
\delta\hbox{]}_{q}
   &\mbox {if $ q^{t+1}+ q+1 \leq \delta \leq q^{t+1}+q^{t}-2$;}\\
\hbox{[}n, n-2m\lceil(\delta-1)(1-1/q)\rceil+4m(\alpha+q-1), d \geq
\delta\hbox{]}_{q}  &\mbox{if $q^{t+1}+\alpha  q^{t}+2\leq \delta
\leq   q^{t+1}+(\alpha +1) q^{t}-2 $;}\\
\hbox{[}n, n-2m\lceil(\delta-1)(1-1/q)\rceil+8m(q-1), d \geq
\delta\hbox{]}_{q}
  &\mbox{if $q^{t+1}+(q-1) q^{t}+2 \leq \delta \leq 2q^{t+1}-2q-1 $;}\\
\hbox{[}n, n-2m(2q^{t+1}-2q^{t}-6q+7), d \geq
2q^{t+1}+2q\hbox{]}_{q} &\mbox {if $2q^{t+1}-2q+1 \leq \delta\leq
2q^{t+1}+2q$.}
 \end{array}
\right.$$
 \end{scriptsize}

 %$2q^{t+1}-2q-1-(2q^t-2-1)-2(q-1)-2(q-1)$\\

 (2) The BCH codes $\mathcal{C}(n,q,\delta+1,0)$ have parameters
 \begin{scriptsize}
 $$ \left\{
\begin{array}{lll}
\hbox{[}n, q^{m}-2m\lceil(\delta-1)(1-1/q)\rceil+4m(q-1), d \geq
2\delta\hbox{]}_{q}
   &\mbox {if $ q^{t+1}+ q+1 \leq \delta \leq q^{t+1}+q^{t}-2$;}\\
\hbox{[}n, q^{m}-2m\lceil(\delta-1)(1-1/q)\rceil+4m(\alpha+q-1), d
\geq 2\delta\hbox{]}_{q}  &\mbox{if $q^{t+1}+\alpha  q^{t}+2\leq
\delta
\leq q^{t+1}+(\alpha +1) q^{t}-2 $;}\\
\hbox{[}n, q^{m}-2m\lceil(\delta-1)(1-1/q)\rceil+8m(q-1), d \geq
2\delta\hbox{]}_{q}
  &\mbox{if $q^{t+1}+(q-1) q^{t}+2 \leq \delta \leq 2q^{t+1}-2q-1 $;}\\
\hbox{[}n, q^{m}-2m(q^{t+1}-2q^{t}-6q+7), d \geq
4q^{t+1}+4q\hbox{]}_{q}  &\mbox {if $2q^{t+1}-2q+1 \leq \delta\leq
2q^{t+1}+2q$.}
 \end{array}
\right.$$
 \end{scriptsize}
\end{theorem}
\begin{proof}
With the conclusions of the coset leaders and cardinalities in hand,
it is natural  to obtain the dimension of BCH codes  for  given
designed distance. Since the proof is very  similar to that of
Theorem 3.7 in \cite{liu }, we have it omitted here.
\end{proof}
\noindent{\bf Remark:} The following  theorems about giving the
dimension can be also verified in the similar way to that of Theorem
3.7 in \cite{liu }, then they will be omitted, too.

\subsection{BCH codes of length $n=q^{m}+1$ for $m$ even}

In this subsection,   suppose that $n=q^{m}+1$ with $m=2t\geq 4$.

\begin{theorem}\label{ther3.4} If $x \not \equiv 0 \bmod q$,
 then the following statements hold:

(1) If $1\leq x\leq q^{t}-1$, then $x$ is a coset leader(see
\cite{Ding7,Ding8}).

(2)  If $ q^{t}+ 2 \leq x \leq 2q^{t}-2$,  then $x$ is a coset
leader.

(3) If $x=q^{t}+1, 2q^{t}-1, 2q^{t}+1 ~\hbox{or}~ 2q^{t}+2$, then
$x$ is not a coset leader.
\end{theorem}
\begin{proof}
 See Appendix \ref{pther3.4}.
\end{proof}

To calculate the actual dimension of BCH codes, it still needs to
get the cardinalities of the cosets containing the coset leaders
below.

\begin{lemma}\label{lemm3.5}
 If $1 \leq x \leq  2q^{t}-2$, then $|C_x|=2m$. \end{lemma}

\begin{proof} Seeking a contradiction,
suppose that $|C_x|=k$  with $k<2m$,  it follows that  $x(q^{k}-1)
\equiv 0$. Since $k|2m$ and $m$ is even,  we have  $k=m,
\frac{m}{2}$, $k=\frac{2m}{3}$ or $k \leq \frac{2m}{5}$.

(1) If $k=m, \frac{m}{2}$, then $(q^{k}-1, n)=(q^{m}-1,
q^{m}+1)=1~\hbox{or}~2$. From $1 \leq x \leq 2q^{t}-2<\frac{n}{2}$,
one has $x(q^{k}-1) \not \equiv 0$,  a contradiction.

(2)  If $m\equiv 0\bmod 3$ and $k=\frac{2m}{3}$, then $(q^{k}-1,
n=q^{m}+1)=q^{\frac{m}{3}}+1$. Since $1 \leq x \leq
2q^{t}-2<\frac{n}{q^{\frac{m}{3}}+1}$
$=q^{\frac{2m}{3}}-q^{\frac{m}{3}}+1$, we get  $x(q^{k}-1) \not
\equiv 0 $.

(3)  If $k\leq \frac{2m}{5}$, it can be derived that $1\leq
x(q^{k}-1)\leq (2q^{t}-2)(q^{\frac{2m}{5}}-1) <n$ and $x(q^{k}-1)
\not \equiv 0$.

All the three cases contradict $x(q^{k}-1) \equiv 0 $. Hence, one
can easily know that $|C_x|=2m$ for every $1 \leq x \leq
2q^{t+1}-2q-1$.
\end{proof}

Based on these results above, the dimension of some BCH codes can be
obtained.

\begin{theorem}\label{ther3.6} Suppose that $n=q^m+1$ with $m=2t\geq 4.$ Then the following statements hold:

 (1) The narrow-sense BCH codes $\mathcal{C}(n,q,\delta,1)$ have
parameters
 $$ \left\{
\begin{array}{lll}
\hbox{[}n, n-2m\lceil(\delta-1)(1-1/q)\rceil+2m, d \geq
\delta\hbox{]}_{q}
&\mbox {if $ q^{t}+2 \leq \delta \leq 2q^{t}-2$;}\\
\hbox{[}n, n-4m(q^{t}-q^{t-1}-1), d \geq 2q^{t}+3\hbox{]}_{q} &\mbox
{if $2q^{t}-1\leq \delta\leq
  2q^{t}+3$.}
 \end{array}
\right.$$

(2) The BCH codes $\mathcal{C}(n,q,\delta+1,0)$ have parameters
 $$ \left\{
\begin{array}{lll}
\hbox{[}n, q^m-2m\lceil(\delta-1)(1-1/q)\rceil+2m, d \geq
2\delta\hbox{]}_{q}
   &\mbox {if $ q^{t}+2 \leq \delta \leq 2q^{t}-2$;}\\
\hbox{[}n, q^m-4m(q^{t}-q^{t-1}-1), d \geq 4q^{t}+6\hbox{]}_{q}
&\mbox {if $2q^{t}-1\leq \delta\leq
  2q^{t}+3$.}
 \end{array}
\right.$$
 \end{theorem}

\section{Dimensions  of $\mathcal{C}(q^{2t+1}+1,q,\delta,b)$ with relatively large $\delta$}
Throughout this section, let $n=q^{m}+1$ with $m=2t+1\geq5$. It will
be divided into two subsections by the parity of $q$. In every
subsection, we first determine the first several largest coset
leaders and calculate
 the cardinalities of the cosets containing them.
 Consequently, it is possible to  obtain  the exact dimension of some  BCH codes in terms of
 relatively large designed distances.

\subsection{BCH  codes over $\mathbb{F}_q$ of odd characteristic}\label{sec5}

 In this subsection, let  $q$ be an  odd prime power, we give the following results.
\begin{lemma}\label{lem4.1} Let $n$ and $q$ be given as above.
Consider

$\delta_{1}=\frac{n}{2}$,
$\delta_{2}=\frac{n}{q+1}\cdot\frac{q-1}{2}$,
$\delta_{3}=\delta_{2}- \frac{2\delta_{2}+(q-1)^{2}}{q^2}$ and
$\delta_{4}=\delta_{3}-(q-1)^{2}$.

 If $m=5$, put $\delta_{5}=\delta_{4}-(q-1)$ and $\delta_{6}=\delta_{5}-1$;

if $m\geq7$, let $\delta_{5}=\delta_{4}-(q^{2}-1)(q-1)^{2}$ and
$\delta_{6}=\delta_{5}- (q-1)^{2}$.
 Then $\delta_{i}(i=1,2,\cdots,6)$ are all coset leaders.
\end{lemma}
\begin{proof}See Appendix \ref{plem4.1}. \end{proof}

 To derive the conclusion of Lemma \ref{lem4.2}, we
first introduce  an {\bf Iterative algorithm}, a quite useful
technology  introduced in \cite{liu} to determine the fist several
largest coset leaders modulo $2^m+1$. Adopting it, we can  partition
$I^{(t)}$ into $2^{t-2}$ disjoint subintervals, where
$$I^{(t)}=\left\{
\begin{array}{lll}
[1,q-2] &\mbox {if $t=2;$}\\
\hbox{[}1,(q-1)^{2}q^{2t-5}\hbox{]}&\mbox {if $t\geq 3$}.
 \end{array}
 \right.$$
This is a key to settling the very intractable problem.

\noindent{\bf Iterative Algorithm 1 (IA-1, for short):}

1) If $t=2$, let
 $I_1=I_{2^0}=[1,q-2]=[a_1,b_1]$;

2) If $t=3$, let $I_2=I_{2^1}=[q-1, (q-1)^{2}
q^{2\times1-1}]=[a_2,b_2]$;

3) If $t=4$, consider that

$I_3=I_{2^1+1}=I_1+b_2=[a_1+(q-1)^{2} q, b_1+(q-1)^{2}
q]=[a_3,b_3]$,

$I_4=I_{2^2}=[a_2+b_2, (q-1)^{2} q^{2\times2-1}]=[a_2+(q-1)^{2}q,
(q-1)^{2}q^{3}]=[a_4,b_4]$;

4) Let $t\geq 5$. Suppose that a partition of $I^{(t-1)}$ is given
by

$I^{(t-1)}=[1, (q-1)^{2}q^{2t-7}]$ $=I_1\bigcup I_2 \cdots \bigcup
I_{2^{t-3}}$, where $I_{j}=[a_j,b_j]$.

For $u=2^{t-3}+j$ with $1\leq j\leq 2^{t-3}-1$, consider

$I_{u}=I_{j}+b_{2^{t-3}}$
$=[a_j+(q-1)^{2}q^{2t-7},b_j+(q-1)^{2}q^{2t-7}].$

For $u=2^{t-2}$, let $I_{u}=[a_{2^{t-3}}+b_{2^{t-3}},
(q-1)^{2}q^{2t-5}]$=$[a_{2^{t-2}}, b_{2^{t-2}}].$

Consequently, a partition of $I^{(t)}=[1,(q-1)^{2}q^{2t-5}]$ can be
obtained as
 $$I^{(t)}=(I_{1}\cup I_{2} \cdots
\cup I_{2^{t-3}})\cup I_{2^{t-3}+1} \cdots \cup
I_{2^{t-2}}=I^{(t-1)} \bigcup (\!\!\bigcup
\limits_{u=2^{t-3}+1}^{2^{t-2}}\!\!I_{u})=\bigcup
\limits_{s=1}^{2^{t-2}}I_s.$$

Using  \textit{IA-1} above, we can obtain the following crucial
result.

\begin{lemma}\label{lem4.2} Let $\delta_3$ and $I^{(t)}$ be given as above.
 If $x\in \delta_3+I^{(t)}$, then $x$ is not a coset leader.
\end{lemma}

\begin{proof}

Following the discussions and  notations of the \textit{Iterative
Algorithm} above, one can derive that

(1): $I_{2^j}=\left\{
\begin{array}{lll}
[1,q-2]   &\mbox {if $j=0$};\\
\hbox{[}(q-1)(1-q+\cdots+q^{2j-2}),(q-1)^{2}q^{2j-1}\hbox{]} &\mbox
{if $1\leq j\leq t-2$}.
 \end{array}
 \right.$

(2): When $t\geq 3$ and $s\in [1, 2^{t-2}-1]$, we have the 2-adic
expansion of $s$:
$$s=a_02^0+a_12^1+a_2 2^2+\cdots+a_{t-3}2^{t-3}=(a_0a_1a_2 \cdots  a_{t-3})_{2}.$$

Set $i=i_{s}=\min\{j|a_j=1,0\leq j\leq t-3\}$.
 Whence $I_s$ can be given by
 $$I_s=I_{2^i}+a_{i+1}(q-1)^{2}q^{2i+1}+a_{i+2}(q-1)^{2}q^{2i+3}+\cdots
+a_{t-3}(q-1)^{2}q^{2t-7}$$

Below, we split into two cases according to   $a_{i+1}$.

 {\bf Case A: $a_{i+1}=0$}

If $a_{i+1}=0$, then we have
\begin{eqnarray*}
I_s&=&I_{2^i}+0\times(q-1)^{2}q^{2i+1}+a_{i+2}(q-1)^{2}q^{2i+3}+\cdots +a_{t-3}(q-1)^{2}q^{2t-7}\\
&=&I_{2^i}+(q-1)^{2}q^{2i+3}\underbrace{(a_{i+2}+a_{i+3}q^2+\cdots+a_{t-3}q^{2(t-i-5)})}_{\lambda}\\
&=&I_{2^i}+(q-1)^{2}q^{2i+3}\lambda, \hbox{where}~ 0 \leq \lambda
\leq 1+q^2+\cdots+ q^{2(t-i-5)} <q^{2(t-i-4)}.
 \end{eqnarray*}

\indent Additionally, it is obvious that
$I_{s}=I_{2^{t-2}}=I_{2^{i}}+(q-1)^{2}q^{2i+3}\cdot \lambda$ with
$i=t-2$ and $\lambda=0$. Hence,  for each $s\in [1, 2^{t-2}]$, there
exist two integers $0\leq i=i_s \leq t-2$ and $0 \leq \lambda <
q^{2(t-i-3)}$ such that
$$I_s=I_{_{i,\lambda}}=I_{2^i}+(q-1)^{2}q^{2i+3}\lambda.$$

For a given integer $s \in [1, 2^{t-2}]$,  if  $l \in
I_s=I_{_{i,\lambda}}$, then $x=\delta_3+l$ can be denoted by
$$x=\delta_3+l_0+(q-1)^{2}q^{2i+3}\lambda,$$ where  $0 \leq i \leq
t-2$,  $0 \leq \lambda < q^{2(t-i-4)}$ and $l_0\in I_{2^i}.$

When $k=2t-2i-1$,   we have $k\geq3$ and
 {
 \begin{eqnarray*}
q^kx
&\equiv&y_{_{x,k}}\\
&=& q^kx
-(\frac{q^{k}-1}{2}-q^{k-1}+(q^{k-3}\cdots-q+1)+(q-1)^{2}q\lambda)n\\
&=&\frac{n}{2}-(q^{2t}-q^{2t-1}+\cdots+q^{k+1})+q^{k-1}-(q^{k-3}-\cdots-q+1)+q^{k}l_0-(q-1)^{2}q\lambda.
\end{eqnarray*}}
 Firstly, we study an upper  bound of $y_{_{x,k}}$:
 {
\begin{eqnarray*}
 y_{_{x,k}} &=&
\frac{n}{2}-(q^{2t}-q^{2t-1}+\cdots+q^{k+1})+q^{k-1}-(q^{k-3}-\cdots-q+1)+q^{k}l_0-(q-1)^{2}q\lambda\\
&\leq& \frac{n}{2}-(q^{2t}-q^{2t-1}+\cdots+q^{k+1})+q^{k-1}-(q^{k-3}-\cdots-q+1)+q^{k}l_0\\
 &\leq &\frac{n}{2}-(q^{2t}-q^{2t-1}+\cdots+q^{k+1})+q^{k-1}-(q^{k-3}-\cdots-q+1)+q^{k}(q-1)^{2}q^{2i-1}\\
&=&\frac{n}{2}\!-(q^{2t}\!-q^{2t-1}\!+\!\cdots\!+q^{2t-2i})\!+\!q^{2t-2i-2}\!-\!(q^{2t-2i-4}\!-\!\cdots-q+1)\!+\!(q-1)^{2}q^{2t-2}\\
&\leq &\frac{n}{2}-(q^{2t}-q^{2t-1}+\cdots+q^{4})+q^{2}-1+(q-1)^{2}q^{2t-2}(\hbox{choose}~i=t-2)\\
&=&\frac{n}{2}-q^{2t-1}+(q^{2t-3}-\cdots+q-1)-q(q-1)^{2}.
\end{eqnarray*}}
Next, we will investigate a  lower bound of $y_{_{x,k}}$:

If $i=0$, then $l_0\in[1,q-2]$ and $0\leq \lambda<q^{2(t-3)}$,  we
get that \begin{eqnarray*}
y_{_{x,k}}&=& \frac{1}{2}(q^{2t+1}-q^{2t}\cdots+q)+(l_0-1)q^{2t-1}+q^{2t-2}-(q-1)\lambda\\
 &>&\frac{1}{2}(q^{2t+1}-q^{2t}\cdots+q)+(l_0-1)q^{2t-1}+q^{2t-2}-(q-1)q^{2(t-3)}\\
 &\geq&\frac{1}{2}(q^{2t+1}-q^{2t}\cdots+q)+(1-1)q^{2t-1}+q^{2t-2}-(q-1)q^{2(t-3)}\\
 &=&\frac{1}{2}(q^{2t+1}-q^{2t}+q^{2t-1}+q^{2t-2}+(q^{2t-3}-\cdots+q^{3}-q^{2}))-(q-1)q^{2(t-3)}.
\end{eqnarray*}

 If $t\geq 3$ and $1 \leq i\leq t-2$, then $l_0\in
I_{2^i}=[b_{i},e_{i}]=[(q-1)(q^{2i-2}\cdots-q+1),(q-1)^{2}q^{2i-1}],$
one can infer that
  { \begin{eqnarray*}
y_{_{x,k}}
&=& \frac{n}{2}-(q^{2t}-q^{2t-1}\cdots+q^{k+1})+q^{k-1}-(q^{k-3}\cdots-q+1)+q^{k}l_0-(q-1)^{2}q\lambda\\
 &>&\frac{n}{2}\!-\!(q^{2t}\!-\!q^{2t-1}\cdots\!+\!q^{k+1})\!+\!q^{k-1}\!-\!(q^{k-3}\cdots-q+1)+q^{k}l_0\!-\!(q-1)^{2}q\times q^{2(t-i-4)}\\
 &\geq&\frac{n}{2}\!-\!(q^{2t}\!-\!q^{2t-1}\cdots\!+q^{k+1})\!+q^{k-1}\!-\!(q^{k-3}\cdots\!-q+1)+q^{k}b_i\!-\!(q-1)^{2}q\!\times q^{2(t-i-4)}\\
&=&\frac{n}{2}-q^{2t}+q^{2t-1}-(q^{2t-3}-\cdots+q^{k}-q^{k-1})-(q^{k-3}\cdots-q+1)-(q-1)^{2}q^{k-6}\\
&\geq&\frac{n}{2}-q^{2t}+q^{2t-1}-(q^{2t-3}-\cdots+q^{3}-q^{3-1})-(q^{3-3}\cdots-q+1)-(q-1)^{2}q^{3-6}\\
&=&\frac{n}{2}-q^{2t}+q^{2t-1}-(q^{2t-3}-\cdots+q-1)+q-2-(q-1)^{2}q^{-3}.
\end{eqnarray*}}
Notice that \begin{eqnarray*}
\lceil\frac{(q-1)n/2-\delta_3}{q}\rceil&=&\frac{n}{2}-q^{2t}+q^{2t-1}-(q^{2t-3}-\cdots+q-1)\hbox{~and}\\
\frac{(q-1)n/2+\delta_3}{q}&=&\frac{n}{2}-q^{2t-1}+(q^{2t-3}-\cdots+q-1).
\end{eqnarray*}
It is easy to observe that
$\lceil\frac{(q-1)n/2-\delta_3}{q}\rceil<y_{_{x,k}}<\frac{(q-1)n/2+\delta_3}{q}$
for $k=2t-2i-1$.

If $\lceil\frac{(q-1)n/2-\delta_3}{q}\rceil<y_{_{x,k}}<
\lfloor\frac{(q-1)n}{2q}\rfloor$, then
$\frac{(q-1)n}{2}-\delta_3<qy_{_{x,k}}< \frac{(q-1)n}{2}$, we can
infer there exists a $j_{_{x,k}}=\frac{(q-1)n}{2}-qy_{_{x,k}}\in
C_{y_{_{x,k}}}$ such that $0< j_{_{x,k}}<\delta_3$.

 If $\lceil\frac{(q-1)n}{2q}\rceil\leq y_{_{x,k}}<\frac{(q-1)n/2+\delta_3}{q}$, then $\frac{(q-1)n}{2} \leq
qy_{_{x,k}}<\frac{(q-1)n}{2}+\delta_3$,  it follows that there
exists  a $j_{_{x,k}}= qy_{_{x,k}}-\frac{(q-1)n}{2}\in
C_{y_{_{x,k}}}$ such that $0<j_{_{x,k}}<\delta_3$.

 {\bf Case B: $a_{i+1}=1$}

For given $s \in [1, 2^{t-2}]$ and  $x\in \delta_3+I_s$,  if
$a_{i+1}=1$, we still choose $k=2t-2i-1$. Similar to the case for
$a_{i+1}=0$, it then can be derived   that
$\frac{(q+1)n/2-\delta_3}{q}<y_{_{x,k}}\leq
\lfloor\frac{(q+1)n/2+\delta_3}{q}\rfloor.$ Also, we shall obtain
that there  exists  either
 $j_{_{x,k}}=\frac{(q+1)n}{2}-qy_{_{x,k}}$ or $j_{_{x,k}}=qy_{_{x,k}}-\frac{(q+1)n}{2}$
   such that $j_{_{x,k}}\in C_{y_{_{x,k}}}$ and $0<j_{_{x,k}}<\delta_3$.

Combining Cases A and B,  for every $x\in [\delta_3+1, \delta_3
+(q-1)^{2}q^{2t-5}]$,  there exists an integer $j_{_{x,k}}\in
[0,\delta_3-1]$
 such that $j_{_{x,k}}\in C_x$, i.e.,  $x$ is not a coset leader. This completes the proof of
 this lemma.
\end{proof}

\begin{lemma}\label{lem4.3}
 Let $\delta_{1}$, $\delta_{2}$ and $\delta_{3}$ be given as
above, if $x>\delta_{3}$ and $x\not=\delta_{1}$ or $\delta_{2}$,
then $x$ is not a coset leader.
\end{lemma}
\begin{proof}   To attain the desired conclusion, it suffices to verify that there exists
$y_{_{x,k}}<x$ or $n-y_{_{x,k}}<x$ for some $k \in [0, m-1]$ and
$x>\delta_{3}$ with $x\not=\delta_{1}$ or $\delta_{2}$. We give the
following:

(1): If $x\in [\frac{n}{2}+1,n-1]$, then $x<n<2x$. Clearly,
$n-y_{_{x,0}}=n-x<x$.

(2): If $x\in [\lceil\frac{(q-1)n}{2q}\rceil,\frac{n}{2}-1]$, then
$(q-1)x<\frac{(q-1)n}{2}<qx$, it follows that
$y_{_{x,1}}-\frac{(q-1)n}{2}=qx-\frac{(q-1)n}{2}<x$.

(3): If
$x\in[\delta_2+1=\frac{n}{q+1}\cdot\frac{q-1}{2}+1,\lfloor\frac{(q-1)n}{2q}\rfloor]$,
then $qx<\frac{(q-1)n}{2}<(q+1)x$, whence
$n-y_{_{x,1}}=\frac{(q-1)n}{2}-qx<x$.

(4): If $x\in [\lceil\frac{(q-1)^{2}n}{2q^{2}}
\rceil,\delta_2-1=\frac{n}{q+1}\cdot\frac{q-1}{2}-1]$, we have
$(q^2-1)x<\frac{(q-1)^{2}n}{2}<q^2x$, one can then infer that
$y_{_{x,2}}=q^2x-\frac{(q-1)^{2}n}{2}<x$.

(5): If $x\in [\lceil \frac{(q-1)^{2}n}{2(q^2+1)} \rceil,
\lfloor\frac{(q-1)^{2}n}{2q^{2}} \rfloor]$, we have
$q^2x<\frac{(q-1)^{2}n}{2}<(q^2+1)x$ and
$n-y_{_{x,2}}=\frac{(q-1)^{2}n}{2}-q^2x<x$.

(6): If $x\in [\lceil\frac{(q^2-1)(q-1)^{2}n}{2q^4} \rceil,\lfloor
\frac{(q-1)^{2}n}{2(q^2+1)} \rfloor]$, we get that
$(q^4-1)x<\frac{(q^2-1)(q-1)^{2}n}{2}<q^4x$, and then
  $y_{_{x,4}}=q^4x-\frac{(q^2-1)(q-1)^{2}n}{2}<x$.

 Summarizing the previous six cases,  $x$ is  not a
coset leader for $x\in [\lceil \frac{(q^2-1)(q-1)^2n}{2q^4}  \rceil,
n-1] \setminus \{\delta_1,\delta_2\}$.
 It then follows from the lemma above  that $x$ is not a coset leader for
 all $x$ with
 $x\in \delta_3+I^{(t)}\bigcup [\lceil  \frac{(q^2-1)(q-1)^2n}{2q^4} \rceil,
n-1] \setminus \{\delta_1,\delta_2\}.$

Notice that for $t=2$,

\begin{eqnarray*} \lceil  \frac{(q^2-1)(q-1)^2n}{2q^4}
\rceil
&=&\frac{1}{2}(q^{5}-2q^{4}+2q^{2}-q)+1\\
&\leq&\frac{1}{2}(q^{5}-2q^{4}+2q^{2}-q)+1+\frac{q+1-4}{2}\\
&=&\delta_3+(q-2)+1,
\end{eqnarray*}

and for $t>2$, we have
\begin{eqnarray*} \lceil  \frac{(q^2-1)(q-1)^2n}{2q^4}
\rceil
&=&\frac{1}{2}(q^{2t+1}-2q^{2t}+2q^{2t-2}-q^{2t-3})+1\\
&<&\delta_3+(q-1)^{2}q^{2t-5}.
\end{eqnarray*}
It is easy to know ($\delta_3+I^{(t)}) \bigcup [\lceil
\frac{(q^2-1)(q-1)^2n}{2q^4} \rceil, n-1]=[\delta_3+1,n-1]$ and then
the desired conclusion can be obtained.
\end{proof}

\begin{theorem}\label{the4.4} For given $n$ and $q$, let $\delta_{1}, \delta_{2}, \delta_{3}$, $\delta_{4}$,
$\delta_{5}$ and $\delta_{6}$ be  given as above, then they are the
first, second, third, fourth, fifth and sixth largest coset leaders,
respectively.
 \end{theorem}

\begin{proof} From Lemmas \ref{lem4.1}-\ref{lem4.3},  it
naturally  follows that $\delta_1$, $\delta_2$ and $\delta_3$ are
the first, second and third largest coset leaders,  respectively. We
now proceed to  derive the following results.

 For  $ 1\leq i\leq \frac{(q-1)^2}{2}-1$, \begin{eqnarray*}
q^{2t-1}(\delta_3-i)
&\equiv&q^{2t-1}(\delta_3-i)-(\frac{q^{2t-1}-1}{2}-q^{2t-2}+(q^{2t-4}\cdots-q+1))n\\
               &=&\frac{n}{2}-q^{2t}+q^{2t-2}-(q^{2t-4}-\cdots-q+1)-q^{2t-1}i\\
               &\leq&\frac{n}{2}-q^{2t}-q^{2t-1}+q^{2t-2}-(q^{2t-4}-\cdots-q+1)<\delta_3-i.
               \end{eqnarray*}

For  $\frac{(q-1)^2}{2}\leq i\leq (q-1)^2-1=q^2-2q$,
\begin{eqnarray*}
-q^{2t-1}(\delta_3-i)
 &\equiv &(\frac{q^{2t-1}-1}{2}-q^{2t-2}+(q^{2t-4}\cdots-q+1))n-q^{2t-1}(\delta_3-i)\\
               &=&q^{2t-1}i-(\frac{n}{2}-q^{2t}+q^{2t-2}-(q^{2t-4}-\cdots-q+1))\\
               &\leq&q^{2t-1}(q^2-2q)-(\frac{n}{2}-q^{2t}+q^{2t-2}-(q^{2t-4}-\cdots-q+1))<\delta_3-i.
 \end{eqnarray*}

This implies that there exists an integer $u_x<x$ such that $u_x \in
 C_x$ for $\delta_4<x<\delta_3$, i.e.,
 $x$ is  not a coset leader for  $\delta_4<x<\delta_3$.
In a similar way, one can easily infer that $x$ is  not a coset
leader for $\delta_5<x<\delta_4$ or $\delta_6<x<\delta_5$.  Hence,
the lemma holds.
\end{proof}

 \begin{lemma}\label{lem4.5}

 If  $\delta_i(i=1,2,\cdots,6)$ are given as above,
  then $|C_{\delta_1}|=1$, $|C_{\delta_2}|=2$ and
  $|C_{\delta_i}|=2m(i=3,4,5,6)$. \end{lemma}

\begin{proof}
From the proof of Lemma \ref{lem4.1}, we have derived that
$C_{\delta_{1}}=\{\delta_{1}\}$ and $C_{\delta_{2}}=\{\delta_{2},
\frac{(q+3)}{q}\delta_{2}\}$. Clearly,  $|C_{\delta_1}|=1$ and
$|C_{\delta_2}|=2$. And  when $x\in \{\delta_3, \delta_4,
\delta_5\}$, $y_{_{x,k}}>x$ for $1\leq k \leq 2m-1$, it immediately
follows that $x(q^{k}-1) \not \equiv 0 $. Hence, one can know that
$|C_x|=2m$.
 \end{proof}

 \begin{theorem}\label{ther}  Suppose that $n=q^{2t+1}+1(t\geq2)$ and $q$ is an odd prime power. Let $\delta_i$ $(i=1,2,\dots 6)$ be
given as above,   then the following statements hold:

 (1) The narrow-sense BCH codes $\mathcal{C}(n,q,\delta,1)$ have
parameters
 $$ \left\{
\begin{array}{lll}
\hbox{[}n, 6m+4, d \geq \delta_5 \hbox{]}_{q}    &\mbox {if $\delta_6+1 \leq \delta \leq \delta_5 $;}\\
\hbox{[}n, 4m+4, d \geq \delta_4 \hbox{]}_{q}    &\mbox {if $\delta_5+1 \leq \delta \leq \delta_4 $;}\\
\hbox{[}n, 2m+4, d \geq \delta_3 \hbox{]}_{q}    &\mbox {if $\delta_4+1 \leq \delta \leq \delta_3 $;}\\
\hbox{[}n, 4, d \geq \delta_2 \hbox{]}_{q}    &\mbox {if $\delta_3+1 \leq \delta \leq \delta_2 $;}\\
\hbox{[}n, 2,  \delta_1=\frac{n}{2} \hbox{]}_{q}       &\mbox {if $\delta_2+1 \leq \delta \leq \delta_1$;}\\
\hbox{[}n, 1, n\hbox{]}_{q}    &\mbox {if $\delta_1+1 \leq \delta
\leq n$.}
 \end{array}
\right.$$

 (2) The BCH codes $\mathcal{C}(n,2,\delta+1,0)$ have parameters
 $$ \left\{
\begin{array}{lll}
\hbox{[}n, 6m+3, d \geq 2\delta_5\hbox{]}_{q}    &\mbox {if $\delta_6+1 \leq \delta \leq \delta_5$;}\\
\hbox{[}n, 4m+3, d \geq 2\delta_4\hbox{]}_{q}    &\mbox {if $\delta_5+1 \leq \delta \leq \delta_4$;}\\
\hbox{[}n, 2m+3, d \geq 2\delta_3\hbox{]}_{q}    &\mbox {if $\delta_4+1 \leq \delta \leq \delta_3$;}\\
\hbox{[}n, 3, d \geq 2\delta_2\hbox{]}_{q}    &\mbox {if $\delta_3+1 \leq \delta \leq \delta_2$;}\\
\hbox{[}n, 1,  2\delta_1=n\hbox{]}_{q}    &\mbox {if $\delta_2+1 \leq \delta \leq \delta_1$.}\\
 \end{array}
\right.$$
\end{theorem}

\noindent{\bf Remark:} Consider $n=q^3+1$.  If $\delta_1$,
$\delta_2$ and $\delta_3$ be given as above. Let
$\delta_4=\delta_3-1$, $\delta_5=\delta_4-(q-1)$ and
$\delta_6=\delta_5-1$.  One can show  that the lemmas and theorems
above hold, too.

\subsection{BCH codes over $\mathbb{F}_q$ of characteristic two } \label{sec3}

In this subsection, let $q=2^{r}\geq4$ and $n=q^{2t+1}+1$ with
$t\geq 2$.

 \begin{lemma}\label{lem4.7}
Suppose that $\delta_{1}=\frac{n}{q+1}\cdot\frac{q}{2}$,
$\delta_{2}=\delta_{1}- \frac{2\delta_{1}+(q-1)q}{q^2}$ and
$\delta_{3}=\delta_{2}-q(q-1)$.

  If $m=5$, let
$\delta_{4}=\delta_{3}-q$ and $\delta_{5}=\delta_{1}-\frac{n}{q+1}$.

If $m\geq7$, let $\delta_{4}=\delta_{3}-q(q^{2}-1)(q-1)$ and
$\delta_{5}=\delta_{4}-q(q-1)$.

 Then $\delta_{1}, \delta_{2},
\delta_{3}$, $\delta_{4}$ and $\delta_{5}$ are all coset leaders.
\end{lemma}
\begin{proof}See Appendix \ref{plem4.7}. \end{proof}
To verify Lemma \ref{lem4.8}, we first introduce  another {\bf
Iterative Algorithm} similar to IA-1. Adopting it, one can partition
$J^{(t)}=[1,(q-1)q^{2(t-2)}]$ into $2^{t-2}$ disjoint subintervals.

{\bf Iterative Algorithm 2 (IA-2, for short):}

1) If $t=2$, let
 $J_1=J_{2^0}=[1,(q-1)q^{2\times0}]=[1,q-1]=[a_1,b_1]$;

2) If $t=3$, let $J_2=J_{2^1}=[a_1+(q-1) q^{2\times0}, (q-1)
q^{2\times1}]=[a_2,b_2]$;

3) If $t=4$, consider that

$J_3=J_{2^1+1}=J_1+b_2=[a_1+(q-1) q^{2}, b_1+(q-1)
q^{2}]=[a_3,b_3]$,

$J_4=J_{2^2}=[a_2+b_2, (q-1) q^{2\times2}]=[a_2+(q-1) q^{2}, (q-1)
q^{4}]=[a_4,b_4]$;

4) Let $t\geq 5$,  a partition of $ J^{(t-1)}$ is given by

 $J^{(t-1)}=[1, (q-1)q^{2(t-3)}]$ $=J_1\bigcup J_2 \cdots \bigcup
J_{2^{t-3}}$, where $J_{j}=[a_j,b_j]$.

  For $u=2^{t-3}+j$ with $1\leq j\leq 2^{t-3}-1$, consider

$J_{u}=J_{j}+b_{2^{t-3}}$ $=[a_j+(q-1) q^{2(t-3)},b_j+(q-1)
q^{2(t-3)}].$

For $u=2^{t-2}$, consider $J_{2^{t-2}}=[a_{2^{t-3}}+b_{2^{t-3}},
(q-1) q^{2(t-2)}]$=$[a_{2^{t-2}}, b_{2^{t-2}}].$

Then, a partition of $J^{(t)}=[1, (q-1)q^{2(t-2)}]$ can be obtained
as follows:
 $$J^{(t)}=J^{(t-1)} \bigcup (\!\!\bigcup
\limits_{u=2^{t-3}+1}^{2^{t-2}}\!\!J_{u})
  =(J_{1}\cup J_{2} \cdots
\cup J_{2^{t-3}})\cup J_{2^{t-3}+1} \cdots \cup J_{2^{t-2}}.$$

 \begin{lemma}\label{lem4.8} Let $\delta_2$ be given as above. If $x\in
\delta_2+J^{(t)}=[\delta_2+1, \delta_2+(q-1) q^{2(t-2)}]$, then $x$
is not a coset leader.
\end{lemma}
\begin{proof}
For $0\leq i \leq t-2$, from   \textit{IA-2} above, we can infer
that

(1):$J_{2^j}=\left\{
\begin{array}{lll}
[1,q-1]=\hbox{[}1,(q-1)q^{2i}\hbox{]}   &\mbox {if $j=0$};\\
\hbox{[}1+(q-1)(1+\cdots+q^{2(j-2)}+q^{2(j-1)}),(q-1)q^{2j}\hbox{]}
&\mbox {if $1\leq j \leq t-2$}.
 \end{array}
 \right.$

(2): When $t\geq 3$ and $s\in [1, 2^{t-2}-1]$, we have the 2-adic
expansion of $s$:
$$s=a_02^0+a_12^1+a_2 2^2+\cdots+a_{t-3}2^{t-3}=(a_0a_1a_2 \cdots  a_{t-3})_{2}.$$

Let $i=i_{s}=\min\{j|a_j=1,0\leq j\leq t-3\}$, then

\begin{eqnarray*}
J_s&=&J_{2^i}+a_{i+1}(q-1)q^{2(i+1)}+a_{i+2}(q-1)q^{2(i+2)}+\cdots +a_{t-3}(q-1)q^{2(t-3)}\\
&=&J_{2^i}+(q-1)q^{2(i+1)}\underbrace{(a_{i+1}+a_{i+2}q^2+\cdots+a_{t-3}q^{2(t-i-4)})}_{\lambda}\\
&=&J_{2^i}+(q-1)q^{2(i+1)}\lambda, \hbox{where}~ 0 \leq \lambda \leq
1+q^2+\cdots+ q^{2(t-i-4)}\!<\!q^{2(t-i-3)}.
 \end{eqnarray*}

Clearly,  $J_{2^{t-2}}=J_{2^{t-2}}+(q-1)q^{2(i+1)}\cdot \lambda$
with $\lambda=0$. Hence,  for each $s\in [1, 2^{t-2}]$,  there exist
two integers $0\leq i\leq t-2$ and $0 \leq \lambda < q^{2(t-i-3)}$
such that
$$J_s=J_{_{i,\lambda}}=J_{2^i}+(q-1)q^{2(i+1)}\lambda.$$

Let $x=\delta_2+l$ with $l \in J^{(t)}=[1,(q-1) q^{2(t-2)}]$, and
$J^{(t)}=\bigcup \limits_{s=1}^{2^{t-2}} J_s=J_1\bigcup J_2 \cdots
\bigcup J_{2^{t-2}}$ be a partition of $J^{(t)}$ denoted as above.

For some $s \in [1, 2^{t-2}]$,  if  $l \in
J_s=J_{_{i,\lambda}}=J_{2^i}+(q-1)q^{2(i+1)}\lambda,$ then
$x=\delta_2+l$ can be denoted by
$x=\delta_2+(l_0+(q-1)q^{2(i+1)}\lambda),$ where  $0 \leq i \leq
t-2$,  $0 \leq \lambda < q^{2(t-i-3)}$ and $l_0\in
J_{2^i}=[\frac{q^{2i}+q}{q+1},(q-1)q^{2i}].$

When $k=2t-2i-1$, we have $k\geq 3$ and
\begin{eqnarray*}
q^kx&\equiv&y_{_{x,k}}\\
&=& q^k(\delta_2+l_0+(q-1)q^{2(i+1)}\lambda)-(\frac{1}{2}(q^{k}-q^{k-1}-\frac{q^{k-1}+q}{q+1})+(q-1)\lambda)n\\
&=& \frac{1}{2}(q^{2t+1}-q^{2t}\cdots+q)-q^{k}+q^{k-1}+q^{k}l_0-(q-1)\lambda\\
&=&\frac{1}{2}(q^{2t+1}-q^{2t}\cdots+q)+(l_0-1)q^{2t-2i-1}+q^{2t-2i-2}-(q-1)\lambda,
\end{eqnarray*}

Firstly, we study an upper  bound of $y_{_{x,k}}$:
\begin{eqnarray*} y_{_{x,k}}&=&
\frac{1}{2}(q^{2t+1}-q^{2t}\cdots+q)+(l_0-1)q^{2t-2i-1}+q^{2t-2i-2}-(q-1)\lambda\\
&\leq& \frac{1}{2}(q^{2t+1}-q^{2t}\cdots+q)+(l_0-1)q^{2t-2i-1}+q^{2t-2i-2}\\
 &\leq &\frac{1}{2}(q^{2t+1}-q^{2t}\cdots+q)+[(q-1)q^{2i}-1]q^{2t-2i-1}+q^{2t-2i-2}\\
&=&\frac{1}{2}[q^{2t+1}+q^{2t}-q^{2t-1}-(q^{2t-2}-\cdots+q^{2}-q)]-(q-1)q^{2t-2i-2}.
\end{eqnarray*}

Next, we will investigate a  lower bound of $y_{_{x,k}}$:

 If $i=0$, then $l_0\in[1,q-1]$ and $0\leq \lambda<q^{2(t-3)}$, we have \begin{eqnarray*}
y_{_{x,k}}&=& \frac{1}{2}(q^{2t+1}-q^{2t}\cdots+q)+(l_0-1)q^{2t-1}+q^{2t-2}-(q-1)\lambda\\
 &>&\frac{1}{2}(q^{2t+1}-q^{2t}\cdots+q)+(l_0-1)q^{2t-1}+q^{2t-2}-(q-1)q^{2(t-3)}\\
 &\geq&\frac{1}{2}(q^{2t+1}-q^{2t}\cdots+q)+(1-1)q^{2t-1}+q^{2t-2}-(q-1)q^{2(t-3)}\\
 &=&\frac{1}{2}[q^{2t+1}-q^{2t}+q^{2t-1}+q^{2t-2}+(q^{2t-3}-\cdots+q^{3}-q^{2})]-(q-1)q^{2(t-3)}.
\end{eqnarray*}

 If $t\geq 3$ and $1 \leq i\leq t-2$, then $l_0\in
J_{2^i}=[\frac{q^{2i}+q}{q+1},(q-1)q^{2i}],$ we get that
  \begin{eqnarray*}
y_{_{x,k}}
&=& \frac{1}{2}(q^{2t+1}-q^{2t}\cdots+q)+(l_0-1)q^{2t-1}+q^{2t-2}-(q-1)\lambda\\
 &>&\frac{1}{2}(q^{2t+1}-q^{2t}\cdots+q)+(l_0-1)q^{2t-1}+q^{2t-2}-(q-1)q^{2(t-3)}\\
 &\geq&\frac{1}{2}(q^{2t+1}-q^{2t}\cdots+q)+(\frac{q^{2i}+q}{q+1}-1)q^{2t-1}+q^{2t-2}-(q-1)q^{2(t-3)}
 \end{eqnarray*}
 \begin{eqnarray*}
&=&\frac{1}{2}[q^{2t+1}-q^{2t}+q^{2t-1}+(q^{2t-2}-q^{2t-3}\cdots+q^{2t-2i-2})]\\
& & +(q^{2t-2i-3}\cdots-q^{2}+q)-(q-1)q^{2(t-i-3)}\\
&\geq&\frac{1}{2}[q^{2t+1}-q^{2t}+q^{2t-1}+(q^{2t-2}-q^{2t-3}\cdots+q^{2t-2(t-2)-2})]\\
& & +(q^{2t-2(t-2)-3}\cdots-q^{2}+q)-(q-1)q^{2(t-(t-2)-3)}\\
&=&\frac{1}{2}[q^{2t+1}-q^{2t}+q^{2t-1}+(q^{2t-2}-\cdots-q^{3}+q^{2})]+q-(q-1)q^{-2}.
\end{eqnarray*}

Note that \begin{eqnarray*}
\frac{qn/2-\delta_2}{q}&=&\frac{1}{2}(q^{2t+1}-q^{2t}+q^{2t-1}+(q^{2t-2}-\cdots-q^{3}+q^{2}))-\frac{q}{2}+1\hbox{~and}\\
\frac{qn/2+\delta_2}{q}&=&\frac{1}{2}(q^{2t+1}+q^{2t}-q^{2t-1}-(q^{2t-2}-\cdots+q^{2}-q)).
\end{eqnarray*}
  It is easy to get that
$\frac{qn/2-\delta_2}{q}<y_{_{x,k}}<\frac{qn/2+\delta_2}{q}$.

If $\frac{qn/2-\delta_2}{q}<y_{_{x,k}}\leq \frac{n-1}{2}$, then
$\frac{qn}{2}-\delta_2<qy_{_{x,k}}\leq \frac{q(n-1)}{2}$, we can
infer there exists a $j_{_{x,k}}=\frac{qn}{2}-qy_{_{x,k}}\in
C_{y_{_{x,k}}}$ such that $\frac{q}{2}\leq j_{_{x,k}}<\delta_2$.

 If $\frac{n+1}{2}\leq y_{_{x,k}}<\frac{qn/2+\delta_2}{q}$, then $\frac{q(n+1)}{2} \leq
qy_{_{x,k}}<\frac{qn}{2}+\delta_2$,  one shall know that there
exists a $j_{_{x,k}}= qy_{_{x,k}}-n\equiv qy_{_{x,k}} \in
C_{y_{_{x,k}}}$ such that $\frac{q}{2}\leq j_{_{x,k}}<\delta_2$.

Concluding the previous discussions,  for all $x\in [\delta_2+1,
\delta_2 +2^{2t-4}]$,  there exists an integer  $\frac{q}{2}\leq
j_{_{x,k}}<\delta_2$ such that $j_{_{x,k}}\in C_x$. Therefore,  $x$
is not a coset leader for  $x\in [\delta_2+1, \delta_2 +2^{2t-4}]$.
 \end{proof}

\begin{lemma}\label{lem4.9}
 Let $\delta_{1}$  and $\delta_{2}$ be given as
above, if $x>\delta_{2}$ and $x\not=\delta_{1}$, then $x$ is not a
coset leader.
\end{lemma}
\begin{proof}   To attain the desired conclusion, it suffices to verify there exists
$y_{_{x,k}}<x$ or $n-y_{_{x,k}}<x$ for some $k \in [0, m-1]$ and
$x>\delta_2$ except that $x=\delta_{1}$. We give   our discussions
in five cases:

(1): If $x\in [\frac{n+1}{2},n-1]$, then $x<n<2x$. Obviously,
$n-y_{_{x,0}}=n-x<x$.

(2): If $x\in
[\delta_1+1=\frac{n}{q+1}\cdot\frac{q}{2}+1,\frac{n-1}{2}]$, then
$qx<\frac{qn}{2}<(q+1)x$. It follows that
$n-y_{_{x,1}}=\frac{qn}{2}-qx<x$.

 (3): If $x\in [\lceil\frac{(q-1)n}{2q} \rceil,\delta_1-1=\frac{n}{q+1}\cdot\frac{q}{2}-1]$,
  then $(q^2-1)x<\frac{q(q-1)n}{2}<q^2x$  and $y_{_{x,2}}=q^2x-\frac{q(q-1)n}{2}<x$.

(4): If $x\in [\lceil \frac{q(q-1)n}{2(q^2+1)} \rceil,
\lfloor\frac{(q-1)n}{2q} \rfloor]$, then
$q^2x<\frac{q(q-1)n}{2}<(q^2+1)x$, whence
$y_{_{x,2}}=\frac{q(q-1)n}{2}-q^2x<x$.

(5): If $x\in [\lceil\frac{(q^2-1)(q-1)n}{2q^3} \rceil,\lfloor
\frac{q(q-1)n}{2(q^2+1)} \rfloor]$,
  we have  $(q^4-1)x<\frac{q(q^2-1)(q-1)n}{2}<q^4x$. Hence,
  $y_{_{x,4}}=q^4x-\frac{q(q^2-1)(q-1)n}{2}<x$.

 We then conclude that  $x$ is  not a
coset leader for $x\in [\lceil\frac{(q^2-1)(q-1)n}{2q^3}\rceil$,
$n-1]\setminus \{\delta_1\}$.
 It   follows from Lemma \ref{lem4.8} that $x$ is not a coset leader for
 all $x$ with
 $x\in [\delta_2+1, \delta_2
+(q-1)q^{2(t-2)}]\bigcup [\lceil  \frac{(q^2-1)(q-1)n}{2q^3} \rceil,
n-1] \setminus \{\frac{n}{3}\}.$

Notice that
\begin{eqnarray*} \lceil \frac{(q^2-1)(q-1)n}{2q^3}
\rceil
&=&\frac{1}{2}(q^{2t+1}-q^{2t}-q^{2t-1}+q^{2t-2})+1\\
&<&\delta_2+(q-1)q^{2(t-2)}.
\end{eqnarray*}
The conclusion can be easily derived from the discussions above.
\end{proof}
\begin{theorem}\label{the4.10}
 Let $\delta_{1}, \delta_{2}, \delta_{3}$,
$\delta_{4}$ and $\delta_{5}$ be  given as above, then they are the
first, second, third, fourth and fifth largest coset leaders,
respectively.
 \end{theorem}
\begin{proof} According to Lemmas \ref{lem4.7}-\ref{lem4.9},  it is enough to
show $\delta_1$ and $\delta_2$ are the first and second largest
coset leaders,  respectively.

We can  easily   derive the following.

 For  $ 1\leq i\leq \frac{q(q-1)}{2}-1$, \begin{eqnarray*}
q^{2t-1}(\delta_2-i) &\equiv&q^{2t-1}(\delta_2-i)-\frac{1}{2}(q^{2t-1}-q^{2t-2}-\frac{q(q^{2t-3}+1)}{q+1})n\\
               &=&\frac{1}{2}(q^{2t+1}-q^{2t}-q^{2t-1}+q^{2t-2}+(q^{2t-3}-q^{2t-4}+\cdots+q))-q^{2t-1}i\\
               &\leq&\frac{1}{2}(q^{2t+1}-q^{2t}-3q^{2t-1}+q^{2t-2}+(q^{2t-3}-q^{2t-4}+\cdots+q))<\delta_2-i;
               \end{eqnarray*}

For  $\frac{q(q-1)}{2}\leq i\leq q(q-1)-1$, \begin{eqnarray*}
-q^{2t-1}(\delta_2-i)
 &\equiv &\frac{1}{2}(q^{2t-1}-q^{2t-2}-\frac{q(q^{2t-3}+1)}{q+1})n-q^{2t-1}(\delta_2-i)\\
               &=&q^{2t-1}i-\frac{1}{2}(q^{2t+1}-q^{2t}-q^{2t-1}+q^{2t-2}+(q^{2t-3}-q^{2t-4}+\cdots+q))\\
               &\leq&\frac{1}{2}(q^{2t+1}-q^{2t}-q^{2t-1}-q^{2t-2}-(q^{2t-3}-q^{2t-4}+\cdots+q))<\delta_2-i.
 \end{eqnarray*}
These discussions above imply that there exists an integer $u_x<x$
such that $u_x \in
 C_x$ for $\delta_3<x<\delta_2$, i.e.,
 $x$ is  not a coset leader for  $\delta_3<x<\delta_2$.

Similarly, one can infer that $x$ is  not a coset leader for
$\delta_4<x<\delta_3$ and $\delta_5<x<\delta_4$. It then immediately
follows that $\delta_{1}, \delta_{2}, \delta_{3}$, $\delta_{4}$ and
$\delta_{5}$  are the first, second, third, fourth and fifth largest
coset leaders,  respectively.
\end{proof}
\begin{lemma}\label{lem4.11}

 If  $x\in \{\delta_2,\delta_3, \delta_4, \delta_5\}$,  then $|C_x|=2m$ except $|C_{\delta_5}|=2$ for $m=5$.
\end{lemma}
\begin{proof}
Combining Lemma \ref{lem4.7}, one can similarly obtain the desired
conclusion. The detailed proof is omitted.
 \end{proof}

According to the previous results on the coset leaders and their
cardinalities,  we can give the following result.

 \begin{theorem}\label{theo3.7} Suppose that $n=q^{2t+1}+1(t\geq2)$ and $q$ is a power of 2. Let  $\delta_i$ $(i=1,2,3,4,5)$ be
given as above,   then the following statements hold:

 (1) The narrow-sense BCH codes $\mathcal{C}(n,q,\delta,1)$ have
parameters
 $$ \left\{
\begin{array}{lll}
\hbox{[}n, 6m+3, d \geq \delta_4 \hbox{]}_{q}    &\mbox {if $\delta_5+1 \leq \delta \leq \delta_4 $;}\\
\hbox{[}n, 4m+3, d \geq \delta_3 \hbox{]}_{q}    &\mbox {if $\delta_4+1 \leq \delta \leq \delta_3 $;}\\
\hbox{[}n, 2m+3, d \geq \delta_2 \hbox{]}_{q}    &\mbox {if $\delta_3+1 \leq \delta \leq \delta_2 $;}\\
\hbox{[}n, 3,  \delta_1 \hbox{]}_{q}       &\mbox {if $\delta_2+1 \leq \delta \leq \delta_1 $;}\\
\hbox{[}n, 1, n\hbox{]}_{q}    &\mbox {if $\delta_1+1 \leq \delta
\leq n$.}
 \end{array}
\right.$$

 (2) The BCH codes $\mathcal{C}(n,2,\delta+1,0)$ have parameters
 $$ \left\{
\begin{array}{lll}
\hbox{[}n, 6m+2, d \geq 2\delta_4\hbox{]}_{q}    &\mbox {if $\delta_5+1 \leq \delta \leq \delta_4$;}\\
\hbox{[}n, 4m+2, d \geq 2\delta_3\hbox{]}_{q}    &\mbox {if $\delta_4+1 \leq \delta \leq \delta_3$;}\\
\hbox{[}n, 2m+2, d \geq 2\delta_2\hbox{]}_{q}    &\mbox {if $\delta_3+1 \leq \delta \leq \delta_2$;}\\
\hbox{[}n, 2,  2\delta_1\hbox{]}_{q}    &\mbox {if $\delta_2+1 \leq \delta \leq \delta_1$.}\\
 \end{array}
\right.$$
\end{theorem}

\noindent{\bf Remark:} Consider $n=q^3+1$.  If $\delta_1$ and
$\delta_2$ are given as above. Let $\delta_3=\delta_2-(q-1)^2$,
$\delta_4=\delta_3-(q-2)$ and $\delta_5=\delta_4-1$. Notice that
$|C_{\delta_3}|=2$ and $|C_{\delta_4}|=|C_{\delta_5}|=2m$.  The
dimensions of $\mathcal{C}(n,q,\delta,1)$ and
$\mathcal{C}(n,q,\delta,0)$ can be also similarly given.

\section{ Some results of BCH codes for $n=q^{m}+1$ with $m$ even}

This section dedicates to  the first largest coset leaders modulo
$n=q^m+1$ with $m$ even. Notice that the structure of the
$q$-cyclotomic cosets modulo $n$ is extremely complex as pointed out
in \cite{Ding7}. For investigating the coset leaders, here   some
new functions and sequences are introduced. However, we  only verify
the first and second largest coset leaders for odd $q$, and the
first largest coset leader for even $q$.
 Some   conjectures are proposed,
which  accurately adapt to  all magma examples of $q$-cyclotomic
cosets we have calculated.
\subsection{BCH codes over $\mathbb{F}_{q}$ of odd characteristic}

In this subsection,   let $q$ be an odd prime power. For giving  the
main conclusions,  we first define a function by
$$\Phi(x,q)=\left\{
\begin{array}{lll}
\frac{q-1}{2}\cdot
(q^{2^{x}}-1)(q^{2^{x-1}}-1)\cdots (q^{2^0}-1) &\mbox {if $x\geq 0$};\\
 \frac{q-1}{2} &\mbox {if $x<0$.}
 \end{array}
 \right.$$

 When $m=2^rt+2^{r-1}\geq 6(r\geq2, t\geq1)$ is not a power of 2,   put

$\delta_1=\frac{n}{2}$,
$\delta_2=\frac{n}{q^{2^{r-1}}+1}\cdot\Phi(r-2,q)$,
$\delta_3=\delta_2-2\cdot \frac{\delta_2+ \Phi(r-1,q)}{q^{2^{r}}}$
and

$\delta_4=\left\{
\begin{array}{lll}
\delta_3-q(q-1)^2  &\mbox {if $t=1$;}\\
\delta_3-2  \Phi(r-1,q)   &\mbox {if $t>1$.}
 \end{array}
\right.$\\

  When  $m=2^{r}\geq4$ is a power of 2,    suppose that

$\delta_1=\frac{n}{2}$,
$\delta_2=\frac{n}{q^{2^{r}}+1}\cdot\Phi(r-1,q)=\Phi(r-1,q)$,
$\delta_3=\delta_2-2 \Phi(r-3,q)$ and

$ \delta_4=\left\{
\begin{array}{lll}
\delta_3-1 &\mbox {if $m=4$;}\\
\delta_2-2 q^{2^{r-2}} \Phi(r-4,q)   &\mbox {if $m\geq 8$.}
 \end{array}
\right.$

Particularly, if $m=2$, let $\delta_1=\frac{n}{2}$,
$\delta_2=\Phi(r-1,q)$, $\delta_3=\delta_2-1$ and
$\delta_4=\delta_2-(q-1)$.

We now first present two critical  results of this subsection:
Theorem \ref{the5.1} and Conjecture \ref{con5.2}.   The proof of
Theorem \ref{the5.1}  will be given after Lemma \ref{lem5.3}.

\begin{theorem}\label{the5.1}
 Let $q$ be an odd prime power. Suppose that $\delta_{1}$ and $\delta_{2}$ are given as above.
 Then $\delta_{1}$ and $\delta_{2}$ are the first, second largest coset leaders, respectively.
\end{theorem}

\begin{conj}\label{con5.2}
 Let $q$ be an odd prime power.  Suppose that $\delta_{3}$ and $\delta_{4}$ are given as above.
Then $\delta_{3}$ and $\delta_{4}$ are the third, fourth largest
coset leaders, respectively.
\end{conj}

To verify Theorem \ref{the5.1},  a kind of   sequences  is
introduced below. First of all, we introduce some notations.

  Consider  $a=(a_u,\cdots,a_{1},a_0)$ and
$b=(b_{u},\cdots,b_{1},b_0)$, where $a_i,b_i\in\{1,-1\}$. We define
a lexicographical comparison over these sequences. Set
$c=a-b=(a_u-b_u,\cdots,a_{1}-b_{1},a_0-b_0)$. If $c=(0,0,\cdots,0)$,
it is said that $a=b$. If $a\neq b$,  define
$i=\max\{j|c_{j}\neq0\}$.  We say that $a>b$ if $a_i-b_i>0$, and
$a<b$, otherwise. For instance, $(1,1,-1,1)>(1,-1,-1,1)$.

Define such a sequence $S^r=(s^{r}_{_{2^{r-1}-1}},\cdots$,
$s^{r}_{_{1}},s^{r}_{_{0}})$,
 which can be generated by the following recursive operations:

$S^1=(1)$, $S^2=(1,-1)$, $S^3=(1,-1,-1,1)$, $\cdots,$
$S^{r-1}=(s^{r-1}_{_{2^{r-2}-1}},\cdots,s^{r-1}_{_{1}},s^{r-1}_{_{0}})$,

$S^r=(S^{r-1},-S^{r-1})
=(s^{r-1}_{_{2^{r-2}-1}},\cdots,s^{r-1}_{_{1}},s^{r-1}_{_{0}},-s^{r-1}_{_{2^{r-2}-1}},\cdots,-s^{r-1}_{_{1}},-s^{r-1}_{_{0}}).$

\begin{lemma}\label{lem5.3}
Let $S^r=(s^{r}_{_{2^{r-1}-1}},\cdots , s^{r}_{_{1}},s^{r}_{_{0}})$
be defined  above.
 For $1\leq k\leq 2^{r-1}-1$, denote two  $k$-Left-Shifts   of $S^{r}$
   by $F^{r}_{(k)}=(f^{r,k}_{_{2^{r-1}-1}},\cdots,f^{r,k}_{_{1}},f^{r,k}_{_{0}})$
    and $H^{r}_{(k)}=(h^{r,k}_{_{2^{r-1}-1}},\cdots,h^{r,k}_{_{1}},h^{r,k}_{_{0}})$, respectively. They can be obtained as

$F^{r}_{(k)}=\left\{
\begin{array}{lll}
(s^{r}_{_{2^{r-1}-1-k}},\cdots,s^{r}_{_{1}},s^{r}_{_{0}},-s^{r}_{_{2^{r-1}-1}},-s^{r}_{_{2^{r-1}-2}}, \cdots,-s^{r}_{_{2^{r-1}-k}}) &\mbox {if $s^{r}_{_{2^{r-1}-1-k}}=1$;}\\
(-s^{r}_{_{2^{r-1}-1-k}},\cdots,-s^{r}_{_{1}},-s^{r}_{_{0}},s^{r}_{_{2^{r-1}-1}},s^{r}_{_{2^{r-1}-2}},\cdots,s^{r}_{_{2^{r-1}-k}})&\mbox
{if $s^{r}_{_{2^{r-1}-1-k}}=-1$,}
\end{array}
\right.$

 $H^{r}_{(k)}=\left\{
\begin{array}{lll}
(s^{r}_{_{2^{r-1}-1-k}},\cdots,s^{r}_{_{1}},s^{r}_{_{0}},s^{r}_{_{2^{r-1}-1}},s^{r}_{_{2^{r-1}-2}}, \cdots,s^{r}_{_{2^{r-1}-k}}) &\mbox {if $s^{r}_{_{2^{r-1}-1-k}}=1$;}\\
-(s^{r}_{_{2^{r-1}-1-k}},\cdots,s^{r}_{_{1}},s^{r}_{_{0}},s^{r}_{_{2^{r-1}-1}},s^{r}_{_{2^{r-1}-2}},\cdots,s^{r}_{_{2^{r-1}-k}})&\mbox
{if $s^{r}_{_{2^{r-1}-1-k}}=-1$.}
 \end{array}
\right.$

Hence, $F^{r}_{(k)}\geq S^{r}$ and $H^{r}_{(k)}\geq S^{r}$  for
$r\geq 2$ and  $1\leq k\leq 2^{r-1}-1$.
\end{lemma}
\begin{proof}It can be verified by the mathematical induction. For details, see Appendix \ref{plem5.3}. \end{proof}
For clarity, we provide an example for the case   $r=3$.
\begin{example} If $r=3$,
note that $S^3=(1,-1,-1,1)$. By definition, we have two
$k$-Left-Shifts as follows.

 $F^3_{(1)}=(1,1,-1,1)$ and $H^3_{(1)}=(1,1,-1,-1)$,

  $F^3_{(2)}=(1,-1,1,-1)$ and $H^3_{(2)}=(1,-1,-1,1)$,

  $F^3_{(3)}=(1,-1,1,1)$ and $H^3_{(3)}=(1,1,-1,-1)$.

It is easy to know that $F^3_{(k)}\geq S^3$ and $H^3_{(k)}\geq S^3$
for $k=1,2,3$. \end{example}

With this result above  in hand,  Theorem  \ref{the5.1} can be well
verified below.

\noindent{\it   Proof of Theorem  \ref{the5.1}:}
 To attain the desired conclusion, it suffices to verify that there exists
$y_{_{x,k}}<x$ or $n-y_{_{x,k}}<x$ for some $k \in [0, m-1]$ and
$x>\delta_{2}$ with $x\not=\delta_{1}$ or $\delta_{2}$. We give the
following discussions.

{\bf Case 1}: $m$ has the form $m=2^{r}u+2^{r-1}(r\geq 2, u\geq 1)$.

 {\bf Step 1.1}: We prove that  $\delta_2$ is a coset leader.

Let $S^r=(s^{r}_{_{2^{r-1}-1}},\cdots,s^{r}_{_{1}},s^{r}_{_{0}})$
and
$F^{r}_{(k)}=(f^{r,k}_{_{2^{r-1}-1}},\cdots,f^{r,k}_{_{1}},f^{r,k}_{_{0}})$be
defined as above. Hence,

$F^{r}_{(k)}=\left\{
\begin{array}{lll}
(s^{r}_{_{2^{r-1}-1-k}},\cdots,s^{r}_{_{1}},s^{r}_{_{0}},-s^{r}_{_{2^{r-1}-1}},-s^{r}_{_{2^{r-1}-2}},\cdots,-s^{r}_{_{2^{r-1}-k}}) &\mbox {if $s^{r}_{_{2^{r-1}-1-k}}=1$;}\\
(-s^{r}_{_{2^{r-1}-1-k}},\cdots,-s^{r}_{_{1}},-s^{r}_{_{0}},s^{r}_{_{2^{r-1}-1}},s^{r}_{_{2^{r-1}-2}},\cdots,s^{r}_{_{2^{r-1}-k}})&\mbox
{if $s^{r}_{_{2^{r-1}-1-k}}=-1$.}
\end{array}
\right.$

 Utilizing $S^r=(s^{r}_{_{2^{r-1}-1}},\cdots,s^{r}_{_{1}},s^{r}_{_{0}})$, we get that
 \begin{eqnarray*}
\delta_{2}&=&\frac{n}{q^{2^{r-1}}+1}\cdot\Phi(r-2,q)\\
                 &=&\frac{n}{q^{2^{r-1}}+1}\cdot\frac{q-1}{2}\cdot
(q^{2^{r-2}}-1)(q^{2^{r-3}}-1)\cdots (q^{2^0}-1)\\
 &=&\frac{n}{q^{2^{r-1}}+1}\cdot\frac{q-1}{2}\cdot \sum
 \limits_{t=0}^{2^{r-1}-1}s^{r}_{t}q^t.
\end{eqnarray*}
We then  show $y_{_{\delta_2,k}}-\delta_2\geq 0$ and
$n-y_{_{\delta_2,k}}-\delta_2\geq 0$ by different $k$ for $0\leq k
\leq 2^{r-1}-1$:

(1): If $k=0$, it is trivial that $y_{_{\delta_2,0}}=\delta_2$ and
$n-2\delta_2>0$.

(2): If $1\leq k \leq 2^{r-1}-1$ and $s^{r}_{2^{r-1}-1-k}=1$, one
can derive that
\begin{eqnarray*}
y_{_{\delta_2,k}}
%&=&q^k\delta_2-\frac{\frac{q-1}{2}\cdot \sum \limits_{t=2^{r-1}-k}^{2^{r-1}-1}s^{r}_{t}q^t}{q^{2^{r-1}-k}}n\\
&=&q^k\delta_2-\frac{q-1}{2}\cdot \sum \limits_{t=2^{r-1}-k}^{2^{r-1}-1}s^{r}_{t}q^{t+k-2^{r-1}}n,\\
&=&\frac{n}{q^{2^{r-1}}+1}\cdot\frac{q-1}{2}\cdot (\sum
\limits_{t=0}^{2^{r-1}-1-k}
s^{r}_{t}q^{t+k}- \sum \limits_{t=2^{r-1}-k}^{2^{r-1}-1}s^{r}_{t}q^{t+k-2^{r-1}})\\
&=&\frac{n}{q^{2^{r-1}}+1}\cdot\frac{q-1}{2}\cdot \sum
\limits_{t=0}^{2^{r-1}-1}f^{r,k}_{t}q^t~(\hbox{see Lemma \ref{lem5.3}}),\\
y_{_{\delta_2,k}}-\delta_2&=&\frac{n}{q^{2^{r-1}}+1}\cdot\frac{q-1}{2}\cdot \sum \limits_{t=0}^{2^{r-1}-1}(f^{r,k}_{t}-s^{r}_{t})q^t\geq0~(\hbox{see Lemma \ref{lem5.3}}),\\
n-y_{_{\delta_2,k}}-\delta_2&=&n-\frac{n}{q^{2^{r-1}}+1}\cdot\frac{q-1}{2}\cdot \sum \limits_{t=0}^{2^{r-1}-1}(f^{r,k}_{t}+s^{r}_{t})q^t\\
&>&n-\frac{n}{q^{2^{r-1}}+1}\cdot\frac{q-1}{2}\cdot \sum
\limits_{t=0}^{2^{r-1}-1}2q^t
=n-\frac{n(q^{2^{r-1}}-1)}{q^{2^{r-1}}+1}>0
\end{eqnarray*}

(3): If $1\leq k \leq 2^{r-1}-1$ and $s^{r}_{2^{r-1}-1-k}=-1$, we
infer that
\begin{eqnarray*}
y_{_{\delta_2,k}}
%&=&q^k\delta_2-\frac{\frac{q-1}{2}\cdot \sum \limits_{t=2^{r-1}-k}^{2^{r-1}-1}s^{r}_{t}q^t}{q^{2^{r-1}-k}}n+n\\
&=&q^k\delta_2-\frac{q-1}{2}\cdot \sum \limits_{t=2^{r-1}-k}^{2^{r-1}-1}s^{r}_{t}q^{t+k-2^{r-1}}n+n,\\
&=&\frac{n}{q^{2^{r-1}}+1}\cdot\frac{q-1}{2}\cdot(\sum
\limits_{t=0}^{2^{r-1}-1-k}
s^{r}_{t}q^{t+k}- \sum \limits_{t=2^{r-1}-k}^{2^{r-1}-1}s^{r}_{t}q^{t+k-2^{r-1}})+n\\
&=&n-\frac{n}{q^{2^{r-1}}+1}\cdot\frac{q-1}{2}\cdot \sum \limits_{t=0}^{2^{r-1}-1}f^{r,k}_{t}q^t~(\hbox{see Lemma \ref{lem5.3}}),\\
y_{_{\delta_2,k}}-\delta_2&=&n-\frac{n}{q^{2^{r-1}}+1}\cdot\frac{q-1}{2}\cdot \sum \limits_{t=0}^{2^{r-1}-1}(f^{r,k}_{t}+s^{r}_{t})q^t\geq0~(\hbox{similar to (2)}),\\
n-y_{_{\delta_2,k}}-\delta_2&=&\frac{n}{q^{2^{r-1}}+1}\cdot\frac{q-1}{2}\cdot
\sum
\limits_{t=0}^{2^{r-1}-1}(f^{r,k}_{t}-s^{r}_{t})q^t>0~(\hbox{see
Lemma \ref{lem5.3}}).
\end{eqnarray*}

{\bf Step 1.2}: We verify that  $x$ is not a coset leader  for
$x>\delta_{2}$ with $x\not=\delta_{1}$  through the following
discussions:

(1): If $x\in[\delta_1+1=\frac{n}{2}+1,n-1]$, then $x<n<2x$.
Clearly, $n-y_{_{x,0}}=n-x<x$.

(2): If $x\in[\lceil\frac{\Phi(-1,q)n}{q^{2^0}}
\rceil,\delta_1-1=\frac{n}{2}-1]$, then
$(q^{2^0}-1)x<\Phi(-1,q)n<q^{2^0}x$, whence
$n-y_{_{x,1}}=\frac{(q-1)n}{2}-q^{2^0}x<x$.

(3): If $x\in[\lceil\frac{\Phi(-1,q)n}{q^{2^0}+1}
\rceil,\lfloor\frac{\Phi(-1,q)n}{q^{2^0}} \rfloor]$, then
$q^{2^0}x<\Phi(-1,q)n<(q^{2^0}+1)x$, whence
$n-y_{_{x,1}}=\frac{(q-1)n}{2}-q^{2^0}x<x$.

For each $i\in[1,r-2]$ with $r\geq3$, the following (4) and (5)
hold:

(4): If $x\in
[\lceil\frac{\Phi(i-1,q)n}{q^{2^{i}}}\rceil,\lfloor\frac{\Phi(i-2,q)n}{q^{2^{i-1}}+1}\rfloor]$,
consider that $q^{2^{i}}-1=(q^{2^{i-1}}+1)(q^{2^{i-1}}-1)$ and
$\Phi(i-1,q)=(q^{2^{i-1}}-1)\Phi(i-2,q)$. Then we have
$(q^{2^{i}}-1)x<\Phi(i-1,q)n<q^{2^{i}}x$, it follows that
$y_{_{x,2^{i}}}-\Phi(i-1,q)n=q^{2^{i}}x-\Phi(i-1,q)n<x$.

(5): If
$x\in[\lceil\frac{\Phi(i-1,q)n}{q^{2^{i}}+1}\rceil,\lfloor\frac{\Phi(i-1,q)n}{q^{2^{i}}}\rfloor]$,
then $q^{2^{i}}x<\Phi(i-1,q)n<(q^{2^{i}}+1)x$, whence
$n-y_{_{x,2^{i}}}=\Phi(i-1,q)n-q^{2^{i}}x<x$.

(6): If $x\in
[\lceil\frac{\Phi(r-2,q)n}{q^{2^{r-1}}}\rceil,\lfloor\frac{\Phi(r-3,q)n}{q^{2^{r-2}}+1}\rfloor]$,
then $(q^{2^{r-1}}-1)x<\Phi(r-2,q)n<q^{2^{r-1}}x$, it follows that
$y_{_{x,2^{r-1}}}-\Phi(r-2,q)n=q^{2^{r-1}}x-\Phi(r-2,q)n<x$.

(7): If $x\in[\delta_2+1=
\frac{\Phi(r-2,q)n}{q^{2^{r-1}}+1}+1,\lfloor\frac{\Phi(r-2,q)n}{q^{2^{r-1}}}\rfloor]$,
then $q^{2^{r-1}}x<\Phi(r-2,q)n<(q^{2^{r-1}}+1)x$, whence
$n-y_{_{x,2^{r-1}}}=\Phi(r-2,q)n-q^{2^{r-1}}x<x$.

{\bf Case 2}: $m=2^r$.  Case 2 is similar to Case 1 above, then we
only present the simplified proof.

{\bf Step 2.1}: When $m=2^r$, we have
$\delta_2=\frac{n}{q^{2^{r}}+1}\cdot\Phi(r-1,q)=\Phi(r-1,q)$. In a
very similar way to Step 1.1, one can verify that $\delta_2$ is a
coset leader. The detailed proof is omitted here.

{\bf Step 2.2}: We  show  that  $x$ is not a coset leader for
$x>\delta_{2}$ and $x\not=\delta_{1}$ similar to Step 1.2.

(1-6) are the same to (1-6) in Step 1.2. Below, we give the
remainder of the proof.

(7): If
$x\in[\lceil\frac{\Phi(r-2,q)n}{q^{2^{r-1}}+1}\rceil,\lfloor\frac{\Phi(r-2,q)n}{q^{2^{r-1}}}\rfloor]$,
then $q^{2^{r-1}}x<\Phi(r-2,q)n<(q^{2^{r-1}}+1)x$, whence
$n-y_{_{x,2^{r-1}}}=\Phi(r-2,q)n-q^{2^{r-1}}x<x$.

(8): If
$x\in[\delta_2=\Phi(r-1,q)n=\frac{\Phi(r-1,q)n}{q^{2^{r}}+1},\lfloor\frac{\Phi(r-2,q)n}{q^{2^{r-1}}}\rfloor]$,
then $q^{2^{r-1}}x<\Phi(r-2,q)n<(q^{2^{r-1}}+1)x$, it follows that
$n-y_{_{x,2^{r-1}}}=\Phi(r-2,q)n-q^{2^{r-1}}x<x$.

Combining Cases 1 and 2, Theorem \ref{the5.1} holds.

\begin{lemma}\label{lemm3.4}
Let $q$ be an odd prime power. Suppose that $\delta_{1}$ and
$\delta_{2}$ be given as above. Then $C_{\delta_{1}}=\{\delta_1\}$
and $|C_{\delta_{2}}|=\left\{
\begin{array}{lll}
2^{r}    &\mbox {if $m=2^{r}u+2^{r-1}(r\geq 2,u\geq1)$;}\\
2m      &\mbox {if $m=2^{r}\geq 2$.}\\
 \end{array}
\right.$
\end{lemma}
\begin{proof}
Since $\delta_{1}=\frac{n}{2}$ and
$\frac{n}{2}(q^2-1)=\frac{q^2-1}{2}n\equiv 0$,
 it directly follows that $C_{\delta_{1}}=\{\delta_1\}$.
Note that
$$\delta_{2}(q^{2^r}-1)=\frac{\Phi(r-2,q)n}{q^{2^{r-1}}+1}\cdot(q^{2^r}-1)
=\Phi(r-2,q)(q^{2^{r-1}}+1)n\equiv 0.$$ When $m=2^{r}u+2^{r-1}$,  by
the proof of Steps 1.1,
 we have know that $y_{_{\delta_{2},k}}>\delta_{2}$ for $0\leq k\leq 2^r-1$.
 It naturally follows that $|C_{\delta_{2}}|=2^{r}$.
Similarly, we can easily  get that $|C_{\delta_{2}}|=2^{r+1}=2m$ for
$m=2^r$.
\end{proof}

Now it is sufficient to give the following result.

\begin{theorem}\label{theo3.7} Suppose $n=q^m+1$ and  $q$ is an odd prime power. Let $\delta_i$ $(i=1,2)$ be
given as above,   then the following statements hold:

(1): $m=2^{r}u+2^{r-1}(r\geq2, u\geq1)$

 (1.1) The narrow-sense BCH codes $\mathcal{C}(n,q,\delta,1)$ have
parameters
 $$ \left\{
\begin{array}{lll}
\hbox{[}n, 2+2^r, d \geq \delta_2 \hbox{]}_{q}    &\mbox {if $\delta=\delta_2 $;}\\
\hbox{[}n, 2,  \delta_1 \hbox{]}_{q}       &\mbox {if $\delta_2+1 \leq \delta \leq \delta_1 $;}\\
\hbox{[}n, 1, n\hbox{]}_{q}    &\mbox {if $\delta_1+1 \leq \delta
\leq n$.}
 \end{array}
\right.$$

 (1.2) The BCH codes $\mathcal{C}(n,2,\delta+1,0)$ have parameters

 $$ \left\{
\begin{array}{lll}
\hbox{[}n, 1+2^r, d \geq 2\delta_2 \hbox{]}_{q}    &\mbox {if $\delta=\delta_2 $;}\\
\hbox{[}n, 1,  2\delta_1=n\hbox{]}_{q}       &\mbox {if $\delta_2+1
\leq \delta \leq \delta_1 $.}
 \end{array}
\right.$$

(2): $m=2^r$.

 (2.1) The narrow-sense BCH codes $\mathcal{C}(n,q,\delta,1)$ have
parameters
 $$ \left\{
\begin{array}{lll}
\hbox{[}n, 2+2m, d \geq \delta_2 \hbox{]}_{q}    &\mbox {if $\delta=\delta_2 $;}\\
\hbox{[}n, 2,  \delta_1 \hbox{]}_{q}       &\mbox {if $\delta_2+1 \leq \delta \leq \delta_1 $;}\\
\hbox{[}n, 1, n\hbox{]}_{q}    &\mbox {if $\delta_1+1 \leq \delta
\leq n$.}
 \end{array}
\right.$$

 (2.2) The BCH codes $\mathcal{C}(n,2,\delta+1,0)$ have parameters
 $$ \left\{
\begin{array}{lll}
\hbox{[}n, 1+2m, d \geq 2\delta_2 \hbox{]}_{q}    &\mbox {if $\delta=\delta_2 $;}\\
\hbox{[}n, 1,  2\delta_1=n \hbox{]}_{q}       &\mbox {if $\delta_2+1
\leq \delta \leq \delta_1 $;}
 \end{array}
\right.$$
\end{theorem}

\subsection{BCH codes over $\mathbb{F}_{q}$ of characteristic two}

In this section,  we only give some conclusions with the detailed
proof omitted since the similarity to last subsection.

Define $\Phi'(x,q)=\left\{
\begin{array}{lll}
\frac{q}{2}\cdot
(q^{2^{x}}-1)(q^{2^{x-1}}-1)\cdots (q^{2^0}-1)  &\mbox {if $x\geq 0$};\\
 \frac{q}{2} &\mbox {if $x<0$.}
 \end{array}
 \right.$

 When $m$ is not the power of 2, let  $m=2^{r}t+2^{r-1}(r\geq2, t\geq1)$.
Suppose that

$\delta_1=\frac{n}{q^{2^{r-1}}+1} \cdot \Phi'(r-2,q)$,
$\delta_2=\delta_1-2 \cdot \frac{\delta_1+ \Phi'(r-1,q)}{q^{2^{r}}}$
and

$\delta_3=\left\{
\begin{array}{lll}
\delta_2-(q-1)q^2  &\mbox {if $t=1$;}\\
\delta_2-2\cdot \Phi'(r-1,q) &\mbox {if $t>1$.}
 \end{array}
\right.$
\\

  When  $m=2^{r}\geq4$. Suppose that
$\delta_1=\Phi'(r-1,q)$, $\delta_2=\delta_1-2  \Phi'(r-3,q)$ and

$ \delta_3=\left\{
\begin{array}{lll}
\delta_1-(q-1)(q^2-1)  &\mbox {if $m=4$;}\\
\delta_1-2  q^{2^{r-2}}  \Phi'(r-4,q)   &\mbox {if $m\geq 8$.}
 \end{array}
\right.$

Particularly, if $m=2$, then $\delta_1=\frac{q(q-1)}{2}$,
$\delta_2=\delta_1-(q-1)$ and $\delta_3=\delta_2-1.$

\begin{theorem}\label{lemm3.4}
 Let $q$ be a power of 2 and $\delta_{1}$ be given as above.
 Then $\delta_{1}$ is the first largest coset leader.
\end{theorem}
\begin{proof} This result can be obtained similar to verifying that
 $\delta_{2}$ is the second largest coset leader in last subsection.
\end{proof}

\begin{conj}
 Let $q$ be a power of 2. Suppose that $\delta_{2}$ and $\delta_{3}$ are given as above.
 Then $\delta_{2}$ and $\delta_{3}$ are the second and third largest coset leaders, respectively.
\end{conj}

\noindent {\bf Remark:} Actually, these conclusions in this section
can unify some previous results.
  All  theorems and conjectures  are applicative  for either $q=2$ or odd $m$.
When $q=2$,  it has been verified in \cite{liu } for
$m\not\equiv0\bmod8$. When $m$ is odd (i.e, $m=2t+1$), $m$ actually
 has the form $m=2^{r}t+2^{r-1}$ and $r=1$.
And we have proved these results  in Section 4.

\section{Conclusion}\label{sec6}
By introducing a new  kind of   sequences and extending some
technologies given in \cite{liu}, this article has
 generalized many conclusions in \cite{liu} from binary field to general finite fields.
On one hand,  we have determined the  coset leaders of $C_x$ with
double range of $x$ in \cite{Ding7,Ding8}. On the other hand, we
have derived or  guessed the first several largest coset leaders
modulo $n=q^m+1$. Then one shall naturally calculate  the dimension
of some antiprimitive BCH codes as well as their Bose distances.

Rather, a table is given to show our main conclusions.
\begin{center}{Table 1. Some results on coset leaders   modulo $n=q^m+1$ over $\mathbb{F}_q$}

\begin{tabular}{c|c|c|c|c}
\hline
   &   $m$ & $q$ &  \hbox{coset leaders} &   \\
   \hline
  \hbox{Section 3}  &  \hbox{general} &\hbox{general} & $q^{\lceil \frac{m}{2}\rceil}< x\leq2q^{\lceil \frac{m}{2}\rceil}+2$ & \hbox{\bf Resolved}  \\
     \hline
 &   &\hbox{odd} & $\delta_1,\delta_2,\cdots, \delta_6$ & \\
 \hbox{Section 4} &  \hbox{odd} & &    & \hbox{\bf Resolved}\\
   &   &\hbox{even}& $\delta_1,\delta_2,\cdots, \delta_5$  & \\
  \hline
 &   &  \hbox{odd} & $\delta_1,\delta_2$&\hbox{\bf Resolved}\\
      &   &   & $\delta_3,\delta_4$&\hbox{Conjecture}\\
    \hbox{Section 5}  &  $\hbox{even}$& &\\
  &    &\hbox{even}& $\delta_1$&\hbox{\bf Resolved}\\
      &   &   & $\delta_2,\delta_3$&\hbox{Conjecture}\\
  \hline
\end{tabular}
\end{center}
{\scriptsize Remark: $\delta_i$ is the $i$-th largest coset leader
as given in the given section.
 ``resolved" is to say that the corresponding coset leader has been resolved,
while ``conjecture" denotes that there is just a conjecture for the
given coset leader.}

For these conjectures, their complete solution would wipe out the
difficult problem of determining the largest coset leaders. It is
believed that they also hold and can be verified in the similar way
to Section 4, though  it would be seriously complicated. The
interested reader is sincerely welcome to solve the remaining parts.
In addition,  the new  kind of  sequences and  iterative method
adopted in this paper may be useful for other types of codes. For
instance, solve the Open Problems 45 and 46 on BCH codes  with
projective length proposed in \cite{Ding9}, calculate the dimension
of  the LCD negacyclic codes as discussed  in \cite{Shixin}, and so
on. Many more results on cyclic codes are expected  in future work.

\section*{Acknowledgements}

 This work is supported by National Natural Science Foundation
of China  under Grant Nos.11471011,11801564 and Natural Science
Foundation of Shaanxi Province under Grant No.2017JQ1032.

\appendix
\section{Appendixes}

The following lemma is first provided. It is much useful to some
derivation processes later and can be naturally obtained by
generalizing $2$ to any prime power $q$
 in  Lemma A.1 (\cite{liu}, Appendix).

\begin{lemma}\label{lemma0} Let $f(k)=q^{-k}a+q^{k}b$,  where $a, b$
and $k$ are positive real numbers. If $k_2\geq k_1\geq
\log_{q}\sqrt{a/b}$, then $f(k_2)\geq f(k_1)$.
\end{lemma}

\subsection{The proof of Theorem \ref{ther3.1}}\label{pther3.1}
\begin{proof}  To prove the conclusion, it is enough to verify (2)-(5).

(2) Let $x \in I=[q^{t+1}+ q+1, q^{t+1}+q^{t}-2]$.  To verify (2),
we  then show $y_{_{x,k}}-x \geq 0$ and $n-y_{_{x,k}}-x \geq 0$ in
the following cases.

(2.1): When $k=0,1,2,\cdots,t-1$, we have $0<q^kx<n$. Then
$y_{_{x,k}}=q^kx\geq x$ and
\begin{eqnarray*}
n-y_{_{x,k}}-x&=&q^{2t+1}+1-(q^{k}+1)x\\
&\geq&q^{2t+1}+1-(q^{k}+1)(q^{t+1}+q^{t}-2)\\
&\geq&q^{2t+1}+1-(q^{t-1}+1)(q^{t+1}+q^{t}-2)\\
&=&q^{2t+1}-(q^{2t}+q^{2t-1}+q^{t+1}+q^{t}-2q^{t-1})+3>0.
\end{eqnarray*}

(2.2): When $k=t,t+1$, we derive  $y_{_{x,k}}=q^kx-q^{k-t}n$. It
follows that
\begin{eqnarray*}
y_{_{x,k}}-x&=&(q^{k}-1)x-q^{k-t}n\\
&\geq&(q^{k}-1)(q^{t+1}+ q+1)-q^{k-t}n\\
&=&q^{k}(q+1-q^{-t})-(q^{t+1}+ q+1)\\
&\geq&q^{t}(q+1-q^{-t})-(q^{t+1}+ q+1)\\
&=&q^{t}-q-2>0.
\end{eqnarray*}
\begin{eqnarray*}
n-y_{_{x,k}}-x&=&(q^{k-t}+1)(q^{2t+1}+1)-(q^{k}+1)x\\
&\geq&(q^{k-t}+1)(q^{2t+1}+1)-(q^{k}+1)(q^{t+1}+q^{t}-2)\\
&=&q^{2t+1}-q^{t+1}-q^{t}-q^{k}(q^{t}-2-q^{-t})\\
&\geq&q^{2t+1}-q^{t+1}-q^{t}-q^{t+1}(q^{t}-2-q^{-t})\\
&=&q^{t+1}-q^{t}+q+3>0.
\end{eqnarray*}

(2.3): For each $k=t+2,t+3,\cdots,2t-1$. To determine $y_{_{x,k}}$,
we divide $I=[q^{t+1}+q+1,q^{t+1}+q^{t}-2]$  into $q^{k-1-t}$
disjoint subintervals as follows:

 $I_{_{1,k}}=[q^{t+1}+ q+1,q^{t+1}+q^{2t+1-k}-1]$,

 $I_{_{\lambda,k}}=[q^{t+1}+(\lambda-1)q^{2t+1-k}+1,
 q^{t+1}+\lambda q^{2t+1-k}-1] ~\hbox{for}~ \lambda \in [2, q^{k-t-1}-1]$,

 $I_{_{\lambda,k}}=[q^{t+1}+(\lambda-1)q^{2t+1-k}+1,q^{t+1}+q^{t}-2]~\hbox{for}~ \lambda=q^{k-t-1}$.

Given $k\in[t+2,2t-1]$ and  $\lambda \in [1, q^{k-t-1}]$, if $x\in
I_{_{\lambda,k}}=[l_{_{\lambda,b}},l_{_{\lambda,e}}]$, it can be
inferred that
$$y_{_{x,k}}=q^kx-(q^{k-t}+\lambda-1)n.$$

The further proof is given below.
 Firstly, we verify that
$y_{_{x,k}}-x>0$.

When $\lambda=1$, let  $x\in I_{_{1,k}}$. We then get that
$y_{_{x,k}}=q^kx-q^{k-t}n$ and
 \begin{eqnarray*}
y_{_{x,k}}-x&=&(q^{k}-1)x-q^{k-t}n\\
&\geq&(q^{k}-1)(q^{t+1}+ q+1)-q^{k-t}n\\
&=&q^{k}(q+1-q^{-t})-(q^{t+1}+ q+1)\\
&\geq&q^{t+2}(q+1-q^{-t})-(q^{t+1}+ q+1)\\
&=&q^{t+3}+q^{t+2}-q^{t+1}-q^{2}-q-1>0.
\end{eqnarray*}
When $\lambda\in[2,q^{k-t-1}]$, we have $x\in
I_{_{\lambda,k}}=[l_{_{\lambda,b}},l_{_{\lambda,e}}]=[q^{t+1}+(\lambda-1)q^{2t+1-k}+1,l_{_{\lambda,e}}]$
and  $y_{_{x,k}}=q^kx-(q^{k-t}+\lambda-1)n$, it follows that
\begin{eqnarray*}
y_{_{x,k}}-x
&=&(q^{k}-1)x-(q^{k-t}+\lambda-1)n\\
&\geq&(q^{k}-1)(q^{t+1}+(\lambda-1)q^{2t+1-k}+1)-(q^{k-t}+\lambda-1)(q^{2t+1}+1)\\
&=&q^{k}-q^{k-t}+q^{2t+1-k}-q^{t+1}-(q^{2t+1-k}+1)\lambda\\
&\geq&q^{k}-q^{k-t}+q^{2t+1-k}-q^{t+1}-(q^{2t+1-k}+1)q^{k-t-1}\\
&=&q^{2t+1-k}+q^{k}(1-q^{-t}-q^{-(t+1)})-q^{t}-q^{t+1}\\
&\geq&q^{2t+1-(t+2)}+q^{t+2}(1-q^{-t}-q^{-(t+1)})-q^{t}-q^{t+1}\hbox{~(see Lemma \ref{lemma0})}\\
&=&q^{t+2}-q^{t+1}-q^{t}+q^{t-1}-q^{2}-q>0.
\end{eqnarray*}
Secondly, we show $n-y_{_{x,k}}-x>0$.

For given $\lambda \in [1, q^{k-t-1}-1]$, since $x\in
I_{_{\lambda,k}}=[l_{_{\lambda,b}},l_{_{\lambda,e}}]=[l_{_{\lambda,b}},q^{t+1}+\lambda
q^{2t+1-k}-1]$, one can deduce that

\begin{eqnarray*}
n-y_{_{x,k}}-x
&=&(q^{k-t}+\lambda)(q^{2t+1}+1)-(q^{k}+1)x\\
&\geq&(q^{k-t}+\lambda)(q^{2t+1}+1)-(q^{k}+1)(q^{t+1}+\lambda q^{2t+1-k}-1)\\
&=&q^{k}+q^{k-t}-q^{t+1}+1-(q^{2t+1-k}-1)\lambda\\
&\geq&q^{k}+q^{k-t}-q^{t+1}+1-(q^{2t+1-k}-1)(q^{k-t-1}-1)\\
&=&q^{2t+1-k}+q^{k}(1+q^{-t}+q^{-(t+1)})-q^{t}-q^{t+1}\\
&\geq&q^{2t+1-(t+2)}+q^{t+2}(1+q^{-t}+q^{-(t+1)})-q^{t}-q^{t+1}\hbox{~(see Lemma \ref{lemma0})}\\
&=&q^{t+2}-q^{t+1}-q^{t}+q^{t-1}+q^{2}+q>0.
\end{eqnarray*}

When $\lambda=q^{k-t-1}$, we have $x\in I_{_{q^{k-t-1},k}}$ and
\begin{eqnarray*}
n-y_{_{x,k}}-x&=&(q^{k-t}+q^{k-t-1})(q^{2t+1}+1)-(q^{k}+1)x\\
&\geq&(q^{k-t}+q^{k-t-1})(q^{2t+1}+1)-(q^{k}+1)(q^{t+1}+q^{t}-2)\\
&=&q^{k}(2+q^{-t}+q^{-(t+1)})-q^{t}-q^{t+1}+2\\
&\geq&2q^{t+2}-q^{t+1}-q^{t}+q^{2}+q+2>0.
\end{eqnarray*}

(2.4):  When $k=2t$, we divide
  $I=[q^{t+1}+ q+1,q^{t+1}+q^{t}-2]$  into $q^{t-1}-1$
 subintervals as follows:

 $I_{_{\lambda,k}}=[q^{t+1}+\lambda q+1,
 q^{t+1}+(\lambda+1)q-1] ~\hbox{for}~ \lambda \in [1, q^{t-1}-2]$,

 $I_{_{q^{t-1}-1,k}}=[q^{t+1}+q^{t}-q+1,q^{t+1}+q^{t}-2]$.

For some $\lambda \in [1, q^{t-1}-1]$, if $x\in I_{_{\lambda,k}}$,
we have $y_{_{x,k}}=q^kx-(q^{k-t}+\lambda)n$.

Firstly, we verify that $y_{_{x,k}}-x>0$.
\begin{eqnarray*}
y_{_{x,k}}-x&=&(q^{k}-1)x-(q^{k-t}+\lambda)n\\
&\geq&(q^{k}-1)(q^{t+1}+\lambda q+1)-(q^{k-t}+\lambda)(q^{2t+1}+1)\\
&=&q^{2t}-q^{t}-q^{t+1}-1-(q+1)\lambda\\
&\geq&q^{2t}-q^{t}-q^{t+1}-1-(q+1)(q^{t-1}-1)\\
&=&q^{2t}-2q^{t}-q^{t+1}-q^{t-1}+q>0.
\end{eqnarray*}

Secondly, we show $n-y_{_{x,k}}-x>0$.

For given $\lambda \in [1, q^{t-1}-2]$, since $x\in
I_{_{\lambda,k}}$, one can deduce that
\begin{eqnarray*}
n-y_{_{x,k}}-x
&=&(q^{t}+\lambda+1)(q^{2t+1}+1)-(q^{2t}+1)x\\
&\geq&(q^{t}+\lambda+1)(q^{2t+1}+1)-(q^{2t}+1)(q^{t+1}+(\lambda+1)q-1)\\
&=&q^{2t}+q^{t}-q^{t+1}+2-q-(q-1)\lambda\\
&\geq&q^{2t}+q^{t}-q^{t+1}+2-q-(q-1)(q^{t-1}-2)\\
&=&q^{2t}-q^{t+1}+q^{t-1}+q>0.
\end{eqnarray*}

When $\lambda=q^{t-1}-1$, we have $x\in I_{_{q^{t-1}-1,k}}$ and
\begin{eqnarray*}
n-y_{_{x,k}}-x&=&(q^{t}+q^{t-1})(q^{2t+1}+1)-(q^{2t}+1)x\\
&\geq&(q^{t}+q^{t-1})(q^{2t+1}+1)-(q^{2t}+1)(q^{t+1}+q^{t}-2)\\
&=&2q^{t+2}-q^{t+1}+q^{t-1}+2>0.
\end{eqnarray*}
Summarizing the four cases above, we can deduce that (2) holds.

(3) Let $1\leq \alpha \leq q-2$, if $ q^{t+1}+\alpha  q^{t}+2\leq x
\leq   q^{t+1}+(\alpha +1) q^{t}-2$,
 then $x$ is a coset leader.

(3.1): When $k=0,1,2,\cdots,t-1$, we have $y_{_{x,k}}=q^kx\geq x$,
it follows that
\begin{eqnarray*}
n-y_{_{x,k}}-x&=&q^{2t+1}+1-(q^{k}+1)x\\
&\geq&q^{2t+1}+1-(q^{k}+1)(q^{t+1}+(\alpha +1) q^{t}-2)\\
&\geq&q^{2t+1}+1-(q^{k}+1)(q^{t+1}+((q-2) +1)) q^{t}-2)\\
&\geq&q^{2t+1}+1-(q^{t-1}+1)(q^{t+1}+(q-1) q^{t}-2)\\
&=&q^{2t+1}-2q^{2t}+q^{2t-1}-2q^{t+1}+2q^{t-1}+q^{t}+3>0.
\end{eqnarray*}
(3.2): When $k=t$, we have $y_{_{x,k}}=q^kx-n$, hence
\begin{eqnarray*}
y_{_{x,k}}-x&=&(q^{k}-1)x-n\\
&\geq&(q^{t}-1)(q^{t+1}+\alpha  q^{t}+2)-(q^{2t+1}+1)\\
&=&(q^{2t}-q^{t})\alpha-q^{t+1}+2q^{t}-3\\
&\geq&(q^{2t}-q^{t})\cdot1-q^{t+1}+2q^{t}-3\\
&=&q^{2t}-q^{t+1}+q^{t}-3>0,\\
%\end{eqnarray*}
%\begin{eqnarray*}
n-y_{_{x,k}}-x&=&2n-(q^{k}+1)x\\
&\geq&2(q^{2t+1}+1)-(q^{t}+1)(q^{t+1}+(\alpha +1) q^{t}-2)\\
&\geq&2(q^{2t+1}+1)-(q^{t}+1)(q^{t+1}+((q-2) +1)) q^{t}-2)\\
&=&q^{2t}-2q^{t+1}+3q^{t}+4>0.
\end{eqnarray*}

(3.3): When  $k=t+1$, for given $\alpha \in[1,q-2]$, if $x \in
[q^{t+1}+\alpha  q^{t}+2,  q^{t+1}+(\alpha +1) q^{t}-2]$,
 we have $y_{_{x,k}}=q^kx-(q+\alpha)n$. It follows that
\begin{eqnarray*}
y_{_{x,k}}-x&=&(q^{k}-1)x-(q+\alpha)n\\
&\geq&(q^{t+1}-1)(q^{t+1}+\alpha  q^{t}+2)-(q+\alpha)(q^{2t+1}+1)\\
&=&q^{t+1}-q-2-(q^{t}+1)\alpha\\
&\geq&q^{t+1}-q-2-(q^{t}+1)(q-2)\\
&=&2q^t-2q>0,
\end{eqnarray*}
\begin{eqnarray*}
n-y_{_{x,k}}-x&=&(q+\alpha+1)n-(q^{k}+1)x\\
&\geq&(q+\alpha+1)(q^{2t+1}+1)-(q^{t+1}+1)(q^{t+1}+(\alpha +1) q^{t}-2)\\
&=&q^{t+1}-q^{t}+q+3-(q^{t}-1)\alpha\\
&\geq&q^{t+1}-q^{t}+q+3-(q^{t}-1)(q-2)\\
&=&q^t+2q+1>0.
\end{eqnarray*}
(3.4): For each $k=t+2,t+3,\cdots,2t$ and   $\alpha \in[1,q-2]$, if
$x \in I=[q^{t+1}+\alpha  q^{t}+2,  q^{t+1}+(\alpha +1) q^{t}-2]$,
we divide  $I$  into $q^{k-1-t}$
 subintervals as follows:

 $I_{_{1,k}}=[q^{t+1}+ \alpha  q^{t}+2,q^{t+1}+\alpha  q^{t}+q^{2t+1-k}-1]$,

 $I_{_{\lambda,k}}=[q^{t+1}+\alpha  q^{t}+(\lambda-1)q^{2t+1-k}+1,
 q^{t+1}+\alpha  q^{t}+\lambda q^{2t+1-k}-1] ~\hbox{for}~ \lambda \in [2, q^{k-t-1}-1]$,

 $I_{_{\lambda,k}}=[q^{t+1}+\alpha  q^{t}+(\lambda-1)q^{2t+1-k}+1,q^{t+1}+(\alpha+1)q^{t}-2]~\hbox{for}~ \lambda=q^{k-t-1}$.

For some $\lambda \in [1, q^{k-t-1}]$, if $x\in
I_{_{\lambda,k}}=[l_{_{\lambda,b}},l_{_{\lambda,e}}]$, we have
$$y_{_{x,k}}=q^kx-(q^{k-t}+\alpha q^{k-t-1}+\lambda-1)n.$$

First, we will show $y_{_{x,k}}-x>0$.

If $\lambda=1$ and $x\in I_{_{\lambda,k}}$, we get that
\begin{eqnarray*}
y_{_{x,k}}-x&=&(q^{k}-1)x-(q^{k-t}+\alpha q^{k-t-1})n\\
&\geq&(q^{k}-1)(q^{t+1}+\alpha  q^{t}+2)-(q^{k-t}+\alpha q^{k-t-1})(q^{2t+1}+1)\\
&=&2q^{k}-q^{t+1}-q^{k-t}-2-(q^{k-t-1}+q^{t})\alpha\\
&\geq&2q^{k}-q^{t+1}-q^{k-t}-2-(q^{k-t-1}+q^{t})(q-2)\\
&=&2q^{k}(2-q^{-t}+q^{-t-1})-2q^{t+1}+2q^{t}-2\\
 &\geq&2q^{t+2}(1-q^{-t}+q^{-t-1})-2q^{t+1}+2q^{t}-2\\
&=&2(q^t-1)(q^{2}-q+1)>0.
\end{eqnarray*}

If $\lambda\in[2,q^{k-t-1}]$ and $x\in
I_{_{\lambda,k}}=[l_{_{\lambda,b}},l_{_{\lambda,e}}]=[q^{t+1}+\alpha
q^{t}+(\lambda-1)q^{2t+1-k}+1,l_{_{\lambda,e}}]$, then
\begin{eqnarray*}
&&y_{_{x,k}}-x\\
&=&(q^{k}-1)x-(q^{k-t}+\alpha q^{k-t-1}+\lambda-1)n\\
&\geq&(q^{k}-1)(q^{t+1}+\alpha  q^{t}+(\lambda-1)q^{2t+1-k}+1)-(q^{k-t}+\alpha q^{k-t-1}+\lambda-1)n\\
&=&q^{k}-q^{t+1}-q^{k-t}+q^{2t+1-k}-(q^{k-t-1}+q^{t})\alpha-(q^{2t+1-k}+1)\lambda\\
&\geq&q^{k}-q^{t+1}-q^{k-t}+q^{2t+1-k}-(q^{k-t-1}+q^{t})\alpha-(q^{2t+1-k}+1)q^{k-t-1}\\
&\geq&q^{k}-q^{t+1}-q^{k-t}+q^{2t+1-k}-(q^{k-t-1}+q^{t})(q-2)-(q^{2t+1-k}+1)q^{k-t-1}\\
&=&q^{2t+1-k}+q^{k}(1-2q^{-t}+q^{-t-1})-2q^{t+1}+q^{t}\\
&\geq&q^{2t+1-(t+2)}+q^{t+2}(1-2q^{-t}+q^{-t-1})-2q^{t+1}+q^{t}\hbox{~(see Lemma \ref{lemma0})}\\
&=&q^{t+2}-2q^{t+1}+q^{t}+q^{t-1}-2q^{2}+q>0,
\end{eqnarray*}

Next, we will show $n-y_{_{x,k}}-x>0$.

If $\lambda\in[1,q^{k-t-1}-1]$ and $x\in
I_{_{\lambda,k}}=[l_{_{\lambda,b}},l_{_{\lambda,e}}]=[l_{_{\lambda,b}},q^{t+1}+\alpha
q^{t}+\lambda q^{2t+1-k}-1]$, then
\begin{eqnarray*}
&&n-y_{_{x,k}}-x\\
&=&(q^{k-t}+\alpha q^{k-t-1}+\lambda)n-(q^{k}+1)x\\
&\geq&(q^{k-t}+\alpha q^{k-t-1}+\lambda)(q^{2t+1}+1)-(q^{k}+1)(q^{t+1}+\alpha  q^{t}+\lambda q^{2t+1-k}-1)\\
&=&q^{k}+q^{k-t}-q^{t+1}+1-(q^{t}-q^{k-t-1})\alpha-(q^{2t+1-k}-1)\lambda\\
&\geq&q^{k}+q^{k-t}-q^{t+1}+1-(q^{t}-q^{k-t-1})\alpha-(q^{2t+1-k}-1)(q^{k-t-1}-1)\\
\end{eqnarray*}
\begin{eqnarray*}
&\geq&q^{k}+q^{k-t}-q^{t+1}+1-(q^{t}-q^{k-t-1})(q-2)-(q^{2t+1-k}-1)(q^{k-t-1}-1)\\
&=&q^{2t+1-k}+q^{k}(1+2q^{-t}-q^{-t-1})-2q^{t+1}+q^{t}\\
&\geq&q^{2t+1-(t+2)}+q^{t+2}(1+2q^{-t}-q^{-t-1})-2q^{t+1}+q^{t}\hbox{~(see Lemma \ref{lemma0})}\\
&=&q^{t+2}-2q^{t+1}+q^{t}+q^{t-1}+2q^{2}-q>0.
\end{eqnarray*}

If $\lambda=q^{k-t-1}$ and $x\in I_{_{\lambda,k}}$, then
\begin{eqnarray*}
&&n-y_{_{x,k}}-x\\
&=&(q^{k-t}+\alpha q^{k-t-1}+q^{k-t-1})n-(q^{k}+1)x\\
&\geq&(q^{k-t}+\alpha q^{k-t-1}+q^{k-t-1})(q^{2t+1}+1)-(q^{k}+1)(q^{t+1}+\alpha  q^{t}+ q^{t}-2)\\
&=&2q^{k}+q^{k-t}+q^{k-t-1}-q^{t+1}-q^{t}+2-(q^{t}-q^{k-t-1})\alpha\\
&\geq&2q^{k}+q^{k-t}+q^{k-t-1}-q^{t+1}-q^{t}+2-(q^{t}-q^{k-t-1})(q-2)\\
&=&q^{k}(2+2q^{-t}-q^{-t-1})-2q^{t+1}+q^{t}+2\\
&\geq&q^{2t}(2+2q^{-t}-q^{-t-1})-2q^{t+1}+q^{t}+2\\
&=&2q^{2t}-2q^{t+1}+3q^{t}-q^{t-1}+2>0.
\end{eqnarray*}
Concluding the previous five cases, we can attain the result of (3).

(4) Let $x \in I=[q^{t+1}+(q-1) q^{t}+2, 2q^{t+1}-2q-1]$. To verify
(4), one needs to show $y_{_{x,k}}-x \geq 0$ and $n-y_{_{x,k}}-x
\geq 0$ in the following cases:

 (4.1)  : When $k=0,1,2,\cdots,t-1$, we
have $y_{_{x,k}}=2^kx\geq x$, it follows that
\begin{eqnarray*}
n-y_{_{x,k}}-x&=&n-(q^{k}+1)x\\
&\geq&q^{2t+1}+1-(q^{k}+1)(2q^{t+1}-2q-1)\\
&\geq&q^{2t+1}+1-(q^{t-1}+1)(2q^{t+1}-2q-1)\\
&=&q^{2t+1}-(2q^{2t}+2q^{t+1}-2q^{t}-q^{t-1}-2q)+2>0.
\end{eqnarray*}
(4.2): When $k=t$, we have $y_{_{x,k}}=q^kx-n$. Hence,

\begin{eqnarray*}
y_{_{x,k}}-x&=&(q^{k}-1)x-n\\
&\geq&(q^{t}-1)(q^{t+1}+(q-1) q^{t}+2)-(q^{2t+1}+1)\\
&=&q^{2t+1}-q^{2t}-2q^{t+1}+3q^{t}-3>0,\\
n-y_{_{x,k}}-x&=&2n-(q^{k}+1)x\\
&\geq&2(q^{2t+1}+1)-(q^{t}+1)(2q^{t+1}-2q-1)\\
&=&q^{2t}+2q+3>0.
\end{eqnarray*}

(4.3): When  $k=t+1$ and $x \in [2q^{t+1}-q^{t}+2, 2q^{t+1}-2q-1]$,
 we have $y_{_{x,k}}=q^kx-(2q-1)n$. It then can be inferred that
\begin{eqnarray*}
y_{_{x,k}}-x&=&(q^{k}-1)x-(2q-1)n\\
&\geq&(q^{t+1}-1)(2q^{t+1}-q^{t}+2)-(2q-1)(q^{2t+1}+1)\\
&=&q^t-2q-1>0.\\
n-y_{_{x,k}}-x&=&2qn-(q^{k}+1)x\\
&\geq&2q(q^{2t+1}+1)-(q^{t+1}+1)(2q^{t+1}-2q-1)\\
&=&2q^{t+2}-q^{t+1}+4q+1>0.
\end{eqnarray*}

(4.4): For each $k=t+2,t+3,\cdots,2t-1$, we divide
  $I=[2q^{t+1}-q^{t}+2, 2q^{t+1}-2q-1]$  into $q^{k-1-t}$
 subintervals as follows:

 $I_{_{1,k}}=[2q^{t+1}-q^{t}+2, 2q^{t+1}-q^{t}+q^{2t+1-k}-1]$,

 $I_{_{\lambda,k}}=[2q^{t+1}-q^{t}+(\lambda-1)q^{2t+1-k}+1,
 2q^{t+1}-q^{t}+\lambda q^{2t+1-k}-1] ~\hbox{for}~ \lambda \in [2, q^{k-t-1}-1]$,

 $I_{_{\lambda,k}}=[2q^{t+1}-q^{t}+(\lambda-1)q^{2t+1-k}+1, 2q^{t+1}-2q-1] ~\hbox{for}~ \lambda=q^{k-t-1}$.

For some $\lambda \in [1, q^{k-t-1}]$, if $x\in
I_{_{\lambda,k}}=[l_{_{\lambda,b}},l_{_{\lambda,e}}]$, we have
$$y_{_{x,k}}=q^kx-(2q^{k-t}-q^{k-t-1}+\lambda-1)n.$$

First, we will show $y_{_{x,k}}-x>0$.

If $\lambda=1$ and $x\in I_{_{\lambda,k}}=I_{_{1,k}}$, then
\begin{eqnarray*}
y_{_{x,k}}-x&=&(q^{k}-1)x-(2q^{k-t}-q^{k-t-1})n\\
&\geq&(q^{k}-1)(2q^{t+1}-q^{t}+2)-(2q^{k-t}-q^{k-t-1})(q^{2t+1}+1)\\
&=&q^{k}(2-2q^{-t}+q^{-t-1})-2q^{t+1}+q^{t}-2\\
&\geq&q^{t+2}(2-2q^{-t}+q^{-t-1})-2q^{t+1}+q^{t}-2\\
&=&2q^{t+2}-2q^{t+1}+q^{t}-2q^{2}+q-2>0.
\end{eqnarray*}
If $\lambda\in[2,q^{k-t-1}]$ and $x\in
I_{_{\lambda,k}}=[l_{_{\lambda,b}},l_{_{\lambda,e}}]=[2q^{t+1}-q^{t}+(\lambda-1)q^{2t+1-k}+1,l_{_{\lambda,e}}]$,
we have \begin{eqnarray*} y_{_{x,k}}-x
&=&(q^{k}-1)x-(2q^{k-t}-q^{k-t-1}+\lambda-1)n\\
&\geq&(q^{k}-1)(2q^{t+1}-q^{t}+(\lambda-1)q^{2t+1-k}+1)-(2q^{k-t}-q^{k-t-1}+\lambda-1)n\\
&=&q^{k}-2q^{t+1}-2q^{k-t}+q^{2t+1-k}+q^{k-t-1}+q^{t}-(q^{2t+1-k}+1)\lambda\\
&\geq&q^{k}-2q^{t+1}-2q^{k-t}+q^{2t+1-k}+q^{k-t-1}+q^{t}-(q^{2t+1-k}+1)q^{k-t-1}\\
&=&q^{2t+1-k}+q^{k}(1-2q^{-t})-2q^{t+1}\\
&\geq&q^{2t+1-(t+2)}+q^{t+2}(1-2q^{-t})-2q^{t+1}\hbox{~(see Lemma \ref{lemma0})}\\
&=&q^{t+2}-2q^{t+1}+q^{t-1}-2q^{2}>0.
\end{eqnarray*}
Next, we will show $n-y_{_{x,k}}-x>0$.

If $\lambda\in[1,q^{k-t-1}-1]$ and $x\in
I_{_{\lambda,k}}=[l_{_{\lambda,b}},l_{_{\lambda,e}}]=[l_{_{\lambda,b}},2q^{t+1}-q^{t}+\lambda
q^{2t+1-k}-1]$, then
\begin{eqnarray*}
n-y_{_{x,k}}-x
&=&(2q^{k-t}-q^{k-t-1}+\lambda)n-(q^{k}+1)x\\
&\geq&(2q^{k-t}-q^{k-t-1}+\lambda)(q^{2t+1}+1)-(q^{k}+1)(2q^{t+1}-q^{t}+\lambda q^{2t+1-k}-1)\\
&=&q^{k}+2q^{k-t}-2q^{t+1}+1+q^{t}-q^{k-t-1}-(q^{2t+1-k}-1)\lambda\\
&\geq&q^{k}+2q^{k-t}-2q^{t+1}+1+q^{t}-q^{k-t-1}-(q^{2t+1-k}-1)(q^{k-t-1}-1)\\
&=&q^{2t+1-k}+q^{k}(1+2q^{-t})-2q^{t+1}\\
&\geq&q^{2t+1-(t+2)}+q^{t+2}(1+2q^{-t})-2q^{t+1}\hbox{~(see Lemma \ref{lemma0})}\\
&=&q^{t+2}-2q^{t+1}+q^{t-1}+2q^{2}>0.
\end{eqnarray*}
If $\lambda=q^{k-t-1}$ and $x\in I_{_{\lambda,k}}$, we obtain that
\begin{eqnarray*}
n-y_{_{x,k}}-x
&=&(2q^{k-t}-q^{k-t-1}+\lambda)n-(q^{k}+1)x\\
&\geq&2q^{k-t}(q^{2t+1}+1)-(q^{k}+1)(2q^{t+1}-2q-1)\\
&=&q^{k}(2q+1+2q^{-t})-2q^{t+1}+2q+1\\
&\geq&q^{t+2}(2q+1+2q^{-t})-2q^{t+1}+2q+1\\
&=&2q^{t+3}+q^{t+2}-2q^{t+1}+2q^{2}+2q+1>0.
\end{eqnarray*}

 (4.4):  When $k=2t$, we partition
  $I=[2q^{t+1}-q^{t}+2, 2q^{t+1}-2q-1]$  into $q^{t-1}-2$
 subintervals as follows:

  $I_{_{1,k}}=[2q^{t+1}-q^{t}+2, 2q^{t+1}-q^{t}+q-1]$,

 $I_{_{\lambda,k}}=[2q^{t+1}-q^{t}+(\lambda-1)q+1,
 2q^{t+1}-q^{t}+\lambda q-1] ~\hbox{for}~ \lambda \in [2,
 q^{t-1}-2]$.

For some $\lambda \in [1, q^{t-1}-2]$, if $x\in
I_{_{\lambda,k}}=[l_{_{\lambda,b}},l_{_{\lambda,e}}]$, we derive
that $$y_{_{x,k}}=q^kx-(2q^{t}-q^{t-1}+\lambda-1)n.$$

Firstly, we will show $y_{_{x,k}}-x>0$.

If $\lambda=1$ and $x\in I_{_{\lambda,k}}=I_{_{1,k}}$, then
\begin{eqnarray*}
y_{_{x,k}}-x&=&(q^{k}-1)x-(2q^{t}-q^{t-1})n\\
&\geq&(q^{2t}-1)(2q^{t+1}-q^{t}+2)-(2q^{t}-q^{t-1})(q^{2t+1}+1)\\
&=&2q^{2t}-2q^{t+1}-q^{t}+q^{t-1}-2>0.
\end{eqnarray*}

If $\lambda\in[2,q^{t-1}-2]$ and $x\in
I_{_{\lambda,k}}=[2q^{t+1}-q^{t}+(\lambda-1)q+1,l_{_{\lambda,e}}]$,
we have
\begin{eqnarray*}
y_{_{x,k}}-x
&=&(q^{k}-1)x-(2q^{t}-q^{t-1}+\lambda-1)n\\
&\geq&(q^{2t}-1)(2q^{t+1}-q^{t}+(\lambda-1)q+1)-(2q^{t}-q^{t-1}+\lambda-1)(q^{2t+1}+1)\\
&=&q^{2t}-2q^{t+1}-q^{t}+q+q^{t-1}-(q+1)\lambda\\
&\geq&q^{2t}-2q^{t+1}-q^{t}+q+q^{t-1}-(q+1)(q^{t-1}-2)\\
&=&q^{2t}-2q^{t+1}-2q^{t}+3q+2>0.
\end{eqnarray*}

Next, we show $n-y_{_{x,k}}-x>0$. For any $\lambda\in[1,q^{t-1}-2]$,
if $x\in
I_{_{\lambda,k}}=[l_{_{\lambda,b}},l_{_{\lambda,e}}]=[l_{_{\lambda,b}},2q^{t+1}-q^{t}+\lambda
q-1]$, then
\begin{eqnarray*}
n-y_{_{x,k}}-x
&=&(2q^{t}-q^{t-1}+\lambda)n-(q^{2t}+1)x\\
&\geq&(2q^{t}-q^{t-1}+\lambda)(q^{2t+1}+1)-(q^{2t}+1)(2q^{t+1}-q^{t}+\lambda q-1)\\
&=&q^{2t}-2q^{t+1}+3q^{t}-q^{t-1}+1-(q-1)\lambda\\
&\geq&q^{2t}-2q^{t+1}+3q^{t}-q^{t-1}+1-(q-1)(q^{t-1}-2)\\
&=&q^{2t}-2q^{t+1}+2q^{t}+2q-1>0.
\end{eqnarray*}

According to the four cases above, one can deduce that (4) holds.

 (5)   It is easy to check the following statements:
\begin{eqnarray*}
  -(q^{t+1}+\alpha  q^{t}-1)q^{t+1} &\equiv&  (q+\alpha)n-(q^{t+1}+\alpha q^{t}-1)q^{t+1}\\
   &=&q^{t+1}+q+\alpha<q^{t+1}+\alpha  q^{t}-1,\\
  (q^{t+1}+\alpha  q^{t}+1)q^{t+1} &\equiv&  (q^{t+1}+\alpha q^{t}+1)q^{t+1}-(q+\alpha)n\\
   &=&q^{t+1}-q-\alpha<q^{t+1}+\alpha  q^{t}+1.
\end{eqnarray*}
When $\beta=1 ~\mbox{or}~ 2$ and $1\leq \gamma \leq \beta q-1$, from
$t\geq2$, we can easily get $(\beta q-\gamma)q^{t}\geq
q^{t}>\beta+\gamma$. It follows that
\begin{eqnarray*}
  -(\beta q^{t+1}-\gamma)q^{t} &\equiv&   \beta n-(\beta q^{t+1}-\gamma)q^{t}
   =\gamma q^{t}+\beta<\beta q^{t+1}-\gamma,\\
  (\beta q^{t+1}+\gamma)q^{t} &\equiv&  (\beta q^{t+1}+\gamma)q^{t}-\beta n
   =\gamma q^{t}-\beta<\beta q^{t+1}+\gamma.
\end{eqnarray*}
From definition, the congruence expressions  above  imply that there
exists some integer $y\in[1,x-1]$ satisfying $y\in C_{x}$ for each
$x$ in (5). Hence,  $x$ is not a coset leader and (5)
follows.\end{proof}

\subsection{The proof of Theorem \ref{ther3.4}}\label{pther3.4}
\begin{proof}  Since (1) has been derived by \cite{Ding7,Ding8},
it suffices to prove  (2) and (3).

(2): To verify (2), one only needs to show  $y_{_{x,k}}-x \geq 0$
and $n-y_{_{x,k}}-x \geq 0$ for all $x \in I=[q^{t}+ 2, 2q^{t}-2]$
and $k\in [0,2t-1]$. Next, we will give the proof through the
following subcases by different $k$.

 (2.1): When $k=0,1,2,\cdots,t-1$, it is easy to get
 $x\leq q^kx<n$. Hence,  for each $x$ we have $y_{_{x,k}}=q^kx\geq
x$ and
\begin{eqnarray*}
n-y_{_{x,k}}-x&=&q^{2t}+1-(q^{k}+1)x\\
&\geq&q^{2t}+1-(q^{t-1}+1)x\\
&\geq&q^{2t}+1-(q^{t-1}+1)(2q^{t}-2)\\
&=&(q^{t}-2q^{t-1}-2)q^{t}+2q^{t-1}+3>0.
\end{eqnarray*}

(2.2): When $k=t$, we derive  $y_{_{x,k}}=q^kx-n$, it then follows
that\begin{eqnarray*} y_{_{x,k}}-x&=&(q^{k}-1)x-n
\geq(q^{t}-1)(q^{t}+ 2)-(q^{2t}+1)=q^{t}-3>0,\\
%\end{eqnarray*}
%\begin{eqnarray*}
n-y_{_{x,k}}-x&=&2n-(q^{k}+1)x\geq2(q^{2t}+1)-(q^{t}+1)(2q^{t}-2)=4>0.
\end{eqnarray*}

(2.3): When $k=t+1, t+2, \cdots, 2t-1$, observing the intractability
to determine $y_{_{x,k}}$,  we first partition $I=[q^{t}+ 2,
2q^{t}-2]$ into $q^{k-t}$ disjoint subintervals below.

 $I_{_{\lambda,k}}=[q^{t}+2,q^{t}+\lambda q^{2t-k}-1] ~\hbox{for}~ \lambda=1$,

 $I_{_{\lambda,k}}=[q^{t}+(\lambda-1)q^{2t-k}+1,
 q^{t}+\lambda q^{2t-k}-1] ~\hbox{for}~ \lambda \in [2, q^{k-t}-1]$,

 $I_{_{\lambda,k}}=[q^{t}+(\lambda-1)q^{2t-k}+1,2q^{t}-2]~\hbox{for}~
 \lambda=q^{k-t}$.

Given $k\in[t+1,2t-1]$,  if $x\in
I_{_{\lambda,k}}=[l_{_{\lambda,b}},l_{_{\lambda,e}}]$ for $\lambda
\in [1, q^{k-t}]$, it is not difficult to obtain that
$$y_{_{x,k}}=q^kx-(q^{k-t}+\lambda-1)n.$$

Then we shall further verify  the remainder of the proof. Firstly,
we will show that $y_{_{x,k}}-x>0$.

If $\lambda=1$, we have $x\in
I_{_{\lambda,k}}=I_{_{1,k}}=[q^{t}+2,q^{t}+ q^{2t-k}-1]$,  it
follows that
\begin{eqnarray*}
y_{_{x,k}}-x&=&(q^{k}-1)x-q^{k-t}n\\
&\geq&(q^{k}-1)(q^{t}+2)-q^{k-t}(q^{2t}+1)\\
&=&q^{k}(2-q^{-t})-q^{t}-2\\
&\geq&q^{t+1}(2-q^{-t})-q^{t}-2\\
&=&2q^{t+1}-q^{t}-q-2>0.
\end{eqnarray*}

If $\lambda\in[2,q^{k-t}]$, then $x\in I_{_{\lambda,k}}=
[l_{_{\lambda,b}},l_{_{\lambda,e}}]=
[q^{t}+(\lambda-1)q^{2t-k}+1,l_{_{\lambda,e}}]$, we get that
\begin{eqnarray*}
y_{_{x,k}}-x
&=&(q^{k}-1)x-(q^{k-t}+\lambda-1)n\\
&\geq&(q^{k}-1)(q^{t}+(\lambda-1)q^{2t-k}+1)-(q^{k-t}+\lambda-1)(q^{2t}+1)\\
&=&q^{2t-k}+q^{k}-q^{t}-q^{k-t}-(q^{2t-k}+1)\lambda\\
&\geq&q^{2t-k}+q^{k}-q^{t}-q^{k-t}-(q^{2t-k}+1)q^{k-t}\\
&=&q^{2t-k}+q^{k}(1-2q^{-t})-2q^{t}\\
&\geq&q^{2t-(t+1)}+q^{t+1}(1-2q^{-t})-2q^{t}\hbox{~(see Lemma \ref{lemma0})}\\
&=&q^{t+1}-2q^{t}+q^{t-1}-2q>0.
\end{eqnarray*}

Next, we will show $n-y_{_{x,k}}-x>0$.

If $\lambda\in[1,q^{k-t}-1]$,  then $x\in I_{_{\lambda,k}}=
[l_{_{\lambda,b}},l_{_{\lambda,e}}]= [l_{_{\lambda,b}},q^{t}+\lambda
q^{2t-k}-1]$, we derive that
\begin{eqnarray*}
 n-y_{_{x,k}}-x
&=&(q^{k-t}+\lambda)n-(q^{k}+1)x\\
&\geq&(q^{k-t}+\lambda)(q^{2t}+1)-(q^{k}+1)( q^{t}+\lambda q^{2t-k}-1)\\
&=&q^{k}+q^{k-t}+1-q^{t}-(q^{2t-k}-1)\lambda\\
&\geq&q^{k}+q^{k-t}+1-q^{t}-(q^{2t-k}-1)(q^{k-t}-1)\\
&=&q^{2t-k}+q^{k}(1+2q^{-t})-2q^{t}\\
&\geq&q^{2t-(t+1)}+q^{t+1}(1+2q^{-t})-2q^{t}\hbox{~(see Lemma \ref{lemma0})}\\
&=& q^{t+1}-2q^{t}+q^{t-1}+2q>0.
\end{eqnarray*}

If $\lambda=q^{k-t}$, we have $x\in
I_{_{\lambda,k}}=I_{_{q^{k-t},k}}=[q^{t}+(\lambda-1)q^{2t-k}+1,2q^{t}-2]$.
Then
\begin{eqnarray*}
 n-y_{_{x,k}}-x
&=&2q^{k-t}n-(q^{k}+1)x\\
&\geq&2q^{k-t}(q^{2t}+1)-(q^{k}+1)(2 q^{t}-2)\\
&=&2(q^{k}+q^{k-t}-q^t+1)\\
&\geq&2(q^{t+1}+q^{(t+1)-t}-q^t+1)\\
&=&2(q^{t+1}-q^t+q+1) >0.
\end{eqnarray*}

This completes the proof of (2.3).

(3) It is easily derived that
\begin{eqnarray*}
(q^t+1)q^t&=&q^{2t}+q^t\equiv q^t-1,\\
(2q^t-1)q^{3t}&=&2q^{4t}-q^{3t}\equiv q^t+2,\\
(2q^t+1)q^t&=&2q^{2t}+q^t\equiv q^t-2,\\
2(q^t+1)q^t&=&2q^{2t}+2q^t\equiv 2(q^t-1).
 \end{eqnarray*}
Thus, if $x=q^{t}+1, 2q^{t}-1, 2q^{t}+1 ~\hbox{or}~ 2q^{t+1}+2$,
 it is obvious that there exists an integer less than $x$ in $C_x$.  That is to say,  $x$  is not
a coset leader and then (3) holds.
\end{proof}

\subsection{The proof of Lemma \ref{lem4.1}}\label{plem4.1}

\begin{proof}
It is easy to derive that $C_{\delta_{1}}=\{\delta_{1}\}$ and
$C_{\delta_{2}}=\{\delta_{2}, \frac{q+3}{q}\delta_{2}\}$, which
implies that both $\delta_{1}$ and $\delta_{2}$  are    coset
leaders. The remainder of the proof is to verify that $\delta_{3}$,
$\delta_{4}$, $\delta_{5}$ and $\delta_{6}$ are also coset leaders
by the following steps.

{\it Step 1:} Consider that
$$\delta_{3}=\delta_{2}-
\frac{2\delta_{2}+(q-1)^{2}}{q^2}=\frac{n}{2}-q^{2t}+(q^{2t-2}-q^{2t-3}+\cdots+q^{2}-q).$$
We then show $y_{_{\delta_3,k}}-\delta_3\geq 0$ and
$n-y_{_{\delta_3,k}}-\delta_3\geq 0$ by four cases below.

(2.1): If $k=0$, it is obvious that $y_{_{\delta_3,k}}=\delta_3$ and
$n-y_{_{\delta_3,k}}=n-\delta_3>\delta_3$.

(2.2): If $k=1,2$, we get
\begin{eqnarray*}
y_{_{\delta_3,k}}&=&q^k\delta_3-(\frac{q^{k}-1}{2}-q^{k-1})n\\
                 &=&\frac{n}{2}+(q^{2t-2}-q^{2t-3}+\cdots+q^{2}-q)q^{k}+q^{k-1},
 \end{eqnarray*}
\begin{eqnarray*}
y_{_{\delta_3,k}}-\delta_3&=&q^{2t}+(q^{2t-2}-q^{2t-3}+\cdots+q^{2}-q)(q^{k}-1)+q^{k-1}>0,\\
n-y_{_{\delta_3,k}}-\delta_3&=&q^{2t}-(q^{2t-2}-q^{2t-3}+\cdots+q^{2}-q)(q^{k}+1)-q^{k-1}\\
&\geq&q^{2t}-(q^{2t-2}-q^{2t-3}+\cdots+q^{2}-q)(q^{2}+1)-q^{2-1}\\
&=&(q^{2t-3}+\cdots-q^{2}+q)(q^{2}+1)-q^{2t-2}-q>0.
\end{eqnarray*}

(2.3):  If $k=3,5,\cdots,2t-1$,  then one can deduce that
\begin{eqnarray*}
y_{_{\delta_3,k}}
&=&q^k\delta_3-(\frac{q^{k}-1}{2}-q^{k-1}+(q^{k-3}\cdots-q+1))n\\
                   &=&\frac{n}{2}-(q^{2t}-q^{2t-1}+\cdots+q^{k+1})+q^{k-1}-(q^{k-3}-\cdots-q+1),\\
                y_{_{\delta_3,k}}-\delta_3
&=&q^{2t-1}-2(q^{2t-2}-q^{2t-3}\cdots-q)-1-q^{k}+2q^{k-1}-q^{k-2}\\
&\geq&q^{2t-1}-2(q^{2t-2}-q^{2t-3}\cdots-q)-1-q^{2t-1}+2q^{(2t-1)-1}-q^{(2t-1)-2}\\
&=&q^{2t-3}-2(q^{2t-4}\cdots+q^{2}-q)-1>0,\\
n-y_{_{\delta_3,k}}-\delta_3
&=&2q^{2t}-q^{2t-1}+q^{k}-2q^{k-1}+q^{k-2}+1\\
&\geq&2q^{2t}-q^{2t-1}+q^{2t-1}-2q^{(2t-1)-1}+q^{(2t-1)-2}+1\\
&=&2(q^{2t}-q^{2t-2})+q^{2t-3}+1>0.
\end{eqnarray*}

(2.4):  If $t\geq 3$ and $k=4,6,\cdots,2t$,   we can similarly infer
that
\begin{eqnarray*}
y_{_{\delta_3,k}}&=&q^k\delta_3-(\frac{q^{k}-1}{2}-q^{k-1}+(q^{k-3}\cdots+q-1))n\\
                   &=&\frac{n}{2}+(q^{2t}-q^{2t-1}+\cdots-q^{k+1})+q^{k-1}-(q^{k-3}-\cdots+q-1),\\
y_{_{\delta_3,k}}-\delta_3
&=&2q^{2t}-q^{2t-1}-q^{k}+2q^{k-1}-q^{k-2}+1\\
&\geq&2q^{2t}-q^{2t-1}-q^{2t}+2q^{2t-1}-q^{2t-2}+1\\
&=&q^{2t}+q^{2t-1}-q^{2t-2}+1>0,
\end{eqnarray*}
\begin{eqnarray*}
n-y_{_{\delta_3,k}}-\delta_3
&=&q^{2t-1}-2(q^{2t-2}\cdots+q^2-q)+q^{k}-2q^{k-1}+q^{k-2}-1\\
&\geq&q^{2t-1}-2(q^{2t-2}\cdots+q^2-q)+q^{4}-2q^{4-1}+q^{4-2}-1\\
&=&q^{2t-1}-2(q^{2t-2}\cdots+q^2-q)+q^{4}-2q^{3}+q^{2}-1>0.
\end{eqnarray*}

From the four cases above, one can conclude that $\delta_3$ is a
coset leader.

{\it Step 2:} Note that
$\delta_{4}=\delta_{3}-(q-1)^{2}=\frac{n}{2}-q^{2t}+(q^{2t-2}-q^{2t-3}+\cdots+q^{2}-q)-(q-1)^{2}.$
We can further show $y_{_{\delta_4,k}}-\delta_4\geq 0$ and
$n-y_{_{\delta_4,k}}-\delta_4\geq 0$ in  six cases.

(2.1): If $k=0$, it is easy to know that
$y_{_{\delta_4,k}}=\delta_4$ and
$n-y_{_{\delta_4,k}}=n-\delta_4>\delta_4$.

(2.2): If $k=1,2$, we get that {\small\begin{eqnarray*}
y_{_{\delta_4,k}}
&=&q^k\delta_4-(\frac{q^{k}-1}{2}-q^{k-1})n\\
                 &=&\frac{n}{2}+(q^{2t-2}-q^{2t-3}+\cdots+q^{2})q^{k}-q^{k-1}(q^3-q^2+q-1),\\
y_{_{\delta_4,k}}-\delta_4
&=&q^{2t}+(q^{2t-2}-q^{2t-3}+\cdots+q^{2})(q^{k}-1)-q^{k-1}(q^3-q^2+q-1)+q^2-q+1\\
&\geq&q^{2t}+(q^{2t-2}-q^{2t-3}+\cdots+q^{2})(q^{1}-1)-q^{1-1}(q^3-q^2+q-1)+q^2-q+1\\
&=&q^{2t}+(q^{2t-2}-q^{2t-3}+\cdots-q^{3})(q-1)+(q-1)^2+1>0,
\end{eqnarray*}
\begin{eqnarray*}
n-y_{_{\delta_4,k}}-\delta_4
&=&q^{2t}-(q^{2t-2}-q^{2t-3}+\cdots-q^3)(q^{k}+1)-(q-1)(q^{k}+1)-q^{k-1}\\
&\geq&q^{2t}-(q^{2t-2}-q^{2t-3}+\cdots-q^3)(q^{2}+1)-(q-1)(q^{2}+1)-q^{2-1}\\
&=&q^{2t}-(q^{2t-2}-q^{2t-3}+\cdots-q)(q^{2}+1)+(q-1)^{2}(q^{2}+1)-q\\
&=&q^{2t}-(q^{2t-2}-q^{2t-3}+\cdots-q+1)(q^{2}+1)+((q-1)^{2}+1)(q^{2}+1)-q\\
&=&q^{2t}-\frac{q^{2t-1}+1}{q+1}(q^{2}+1)+((q-1)^{2}+1)(q^{2}+1)-q\\
&=&\frac{q^{2t}-q^{2t-1}-q^2-1}{q+1}(q^{2}+1)+((q-1)^{2}+1)(q^{2}+1)-q>0.
\end{eqnarray*}}
(2.3):  If $t\geq 3$ and $k=3,5,\cdots,2t-3$,  it shall be deduced
{\begin{eqnarray*}
 y_{_{\delta_4,k}}
&=&q^k\delta_4-(\frac{q^{k}-1}{2}-q^{k-1}+(q^{k-3}\cdots-q+1))n\\
                   &=&\frac{n}{2}-(q^{2t}-\cdots-q+1)-q^{k-2}(q^4-2q^3+2q^2-2q+1),\\
y_{_{\delta_4,k}}-\delta_4
&=&q^{2t-1}-2(q^{2t-2}-q^{2t-3}\cdots-q^{3})-q^{2}-q^{k-2}(q^4-2q^3+2q^2-2q+1)\\
&\geq&q^{2t-1}\!-\!2(q^{2t-2}\!-\!q^{2t-3}\cdots-q^{3})\!-\!q^{2}\!-\!q^{(2t-3)\!-\!2}(q^4-2q^3+2q^2-2q+1)\\
&=&2(q^{2t-5}\cdots-q^{2}+q)-q^{2t-5}+q^{2}-2q>0,\\
n-y_{_{\delta_4,k}}-\delta_4
&=&2q^{2t}-q^{2t-1}+(q-1)^2+1+q^{k-2}(q^4-2q^3+2q^2-2q+1)>0.
\end{eqnarray*}}

(2.4):  If $t\geq 3$ and $k=4,6,\cdots,2t-2$, we  then can similarly
obtain that
\begin{eqnarray*}
y_{_{\delta_4,k}}
&=&q^k\delta_4-(\frac{q^{k}-1}{2}-q^{k-1}+(q^{k-3}\cdots+q-1))n\\
                   &=&\frac{n}{2}+(q^{2t}-\cdots-q+1)-q^{k-2}(q^4-2q^3+2q^2-2q+1),
\end{eqnarray*}
\begin{eqnarray*}
y_{_{\delta_4,k}}-\delta_4
&=&2q^{2t}-q^{2t-1}+(q-1)^2+1-q^{k-2}(q^4-2q^3+2q^2-2q+1)\\
&\geq&2q^{2t}-q^{2t-1}+(q-1)^2+1-q^{(2t-2)-2}(q^4-2q^3+2q^2-2q+1)\\
&=&q^{2t}+q^{2t-1}-2q^{2t-2}+2q^{2t-3}-q^{2t-4}+(q-1)^2+1>0,\\
n-y_{_{\delta_4,k}}-\delta_4
&=&q^{2t-1}-2(q^{2t-2}-q^{2t-3}\cdots-q^{3})-q^{2}+q^{k-2}(q^4-2q^3+2q^2-2q+1)\\
&\geq&q^{2t-1}-2(q^{2t-2}-q^{2t-3}\cdots-q^{3})-q^{2}+q^{4-2}(q^4-2q^3+2q^2-2q+1)\\
&=&q^{2t-1}-2(q^{2t-2}-q^{2t-3}\cdots-q^{3})+q^{2}(q^4-2q^3+2q^2-2q)>0.
\end{eqnarray*}

(2.5):  If $k=2t-1$, it then follows that
\begin{eqnarray*}
y_{_{\delta_4,k}}
&=&q^k\delta_4-(\frac{q^{k}-1}{2}-q^{k-1}+(q^{k-3}\cdots-q+1)-1)n\\
                   &=&\frac{n}{2}+q^{2t}-q^{2t-1}+q^{2t-2}-(q^{2t-4}-\cdots-q+1)+1,\\
y_{_{\delta_4,k}}-\delta_4
&=&2q^{2t}-q^{2t-1}-q^{2t-3}+2(q^{2t-3}\cdots-q^2+q)+(q-1)^2>0,\\
n-y_{_{\delta_4,k}}-\delta_4
&=&q^{2t-1}-2q^{2t-2}+q^{2t-3}+(q-1)^2>0.
\end{eqnarray*}

(2.6):  If $k=2t$, then one can get that
\begin{eqnarray*}
y_{_{\delta_4,k}}&=&q^k\delta_4-(\frac{q^{k}-1}{2}-q^{k-1}+(q^{k-3}\cdots-q^2+q)-q+1)n\\
                   &=&\frac{n}{2}-q^{2t}+q^{2t-1}-(q^{2t-3}-\cdots+q-1)+q-2,
\end{eqnarray*}
\begin{eqnarray*}
y_{_{\delta_4,k}}-\delta_4&=&q^{2t-1}-q^{2t-2}+q^{2}-q>0,\\
n-y_{_{\delta_4,k}}-\delta_4
&=&2q^{2t}-q^{2t-1}-q^{2t-2}+2(q^{2t-3}\cdots-q^2+q)+q^{2}-3q+2>0.
\end{eqnarray*}

Concluding the discussions above, one can infer  that $\delta_4$ is
a coset leader.

In the similar way to the proofs of Steps 1 and 2, we can also
derive that both  $\delta_5$ and $\delta_6$ are coset leaders. The
detailed proofs are omitted here.    \end{proof}

\subsection{The proof of Lemma \ref{lem4.7}}\label{plem4.7}

\begin{proof} It is easy to derive that
 $C_{\delta_{1}}=\{\delta_{1}, \frac{q+2}{q}\delta_{1}\}$,
 which implies that $|C_{\delta_{1}}|=2$ and  $\delta_{1}$ is a coset leader.
Next,  we will verify that $\delta_{2}$, $\delta_{3}$, $\delta_{4}$
and $\delta_{5}$ are all also coset leaders by the following
discussions.

{\it Step 1:} Notice that $$\delta_{2}=\delta_{1}-
\frac{2\delta_{1}+(q-1)q}{q^2}=\frac{1}{2}((q^{2t+1}-q^{2t})-(q^{2t-1}-q^{2t-2}+\cdots-q^{2}+q)).$$
We then split into five cases to  show
$y_{_{\delta_2,k}}-\delta_2\geq 0$ and
$n-y_{_{\delta_2,k}}-\delta_2\geq 0$ by different $k$:

(1.1): If $k=0$, it is clear that $y_{_{\delta_2,k}}=\delta_2$ and
$n-y_{_{\delta_2,k}}=n-\delta_2>\delta_2$.

(1.2): If $k=1$,   we easily get that
\begin{eqnarray*}
y_{_{\delta_2,k}}&=&q\delta_2-\frac{q-2}{2}n
                 =\frac{1}{2}((q^{2t+1}-q^{2t})+(q^{2t-1}-q^{2t-2}+\cdots-q^{2})-q+2),\\
y_{_{\delta_2,k}}-\delta_2&=&q^{2t-1}-q^{2t-2}+\cdots-q^{2}+1>0,\\
n-y_{_{\delta_2,k}}-\delta_2&=&q^{2t}+q>0.
\end{eqnarray*}

(1.3): If $k=2$, it is not difficult to  obtain that
\begin{eqnarray*}
y_{_{\delta_2,k}}&=&q^2\delta_2-\frac{q^{2}-q-2}{2}n\\
                 &=&\frac{1}{2}q^{2t+1}+\frac{1}{2}((q^{2t}-q^{2t-1}+q^{2t-2}\cdots-q^{3})-q^{2}+q+2),\\
y_{_{\delta_2,k}}-\delta_2&=&q^{2t}-q^{2}+q+1>0,\\
n-y_{_{\delta_2,k}}-\delta_2&=&q^{2t-1}-q^{2t-2}+\cdots+q^{3}>0.
\end{eqnarray*}

(1.4):  If $k=3,5,\cdots,2t-1$,    we shall infer that
\begin{eqnarray*}
y_{_{\delta_2,k}}&=&q^k\delta_2-\frac{1}{2}(q^{k}-q^{k-1}-\frac{q(q^{k-2}+1)}{q+1})n\\
                 &=&\frac{1}{2}(q^{2t+1}-q^{2t}+q^{2t-1}\cdots-q^{2}+q)-q^{k}+q^{k-1},\\
y_{_{\delta_2,k}}-\delta_2&=&(q^{2t-1}-q^{2t-2}+q^{2t-3}\cdots-q^{2}+q)-(q^{k}-q^{k-1})\\
&\geq&(q^{2t-1}-q^{2t-2}+q^{2t-3}\cdots-q^{2}+q)-(q^{2t-1}-q^{2t-2})\\
&=&q^{2t-3}-q^{2t-4}+q^{2t-5}\cdots-q^{2}+q>0,\\
n-y_{_{\delta_2,k}}-\delta_2&=&q^{2t}+q^{k}-q^{k-1}+1\geq
q^{2t}+(q^{3}-q^{2})+1>0.
\end{eqnarray*}

(1.5):  If $k=4,6,\cdots,2t$,  it can be similarly derived that
\begin{eqnarray*}
y_{_{\delta_2,k}}&=&2^k\delta_2-\frac{1}{2}(q^{k}-q^{k-1}-q^{k-2}+\frac{q(q^{k-3}+1)}{q+1})n+n\\
                 &=&\frac{1}{2}q^{2t+1}+\frac{1}{2}(q^{2t}-q^{2t-1}+\cdots+q^{2}-q)-q^{k}+q^{k-1}+1,
\end{eqnarray*}
\begin{eqnarray*}
y_{_{\delta_2,k}}-\delta_2&=&q^{2t}-(q^{k}-q^{k-1})+1\\
&\geq& q^{2t}-(q^{2t}-q^{2t-1})+1=q^{2t-1}+1>0,\\
n-y_{_{\delta_2,k}}-\delta_2&=&(q^{2t-1}-q^{2t-2}\cdots-q^{2}+q)+q^{k}-q^{k-1}\\
&\geq&(q^{2t-1}-q^{2t-2}\cdots-q^{2}+q)+(q^{4}-q^{3})>0.
\end{eqnarray*}

Collecting  these discussions  above, one can conclude that
$\delta_2$ is a coset leader.

{\it Step 2:} Observe that
$$\delta_{3}=\delta_{2}-q(q-1)=\frac{1}{2}((q^{2t+1}-q^{2t})-(q^{2t-1}-q^{2t-2}+\cdots+q^{3})-(q^{2}-q)).$$
We then give the proof  by the following three cases:

(2.1): If $k=0$, clearly, $y_{_{\delta_3,k}}=\delta_3$ and
$n-y_{_{\delta_3,k}}=n-\delta_3>\delta_3$.

(2.2): If $k=1$,  we have then
\begin{eqnarray*}
y_{_{\delta_3,k}}&=&q\delta_3-\frac{q-2}{2}n\\
                 &=&\frac{1}{2}((q^{2t+1}-q^{2t})\!\!+\!\!(q^{2t-1}-q^{2t-2}+\cdots-q^{2})\!\!-\!\!q+2)\!\!-\!\!q^{2}(q-1),\\
y_{_{\delta_3,k}}-\delta_3&=&((q^{2t-1}\!\!-\!\!q^{2t-2}\!\!+\cdots+q^{3})-q^3)+q^2-q+1>0,\\
n-y_{_{\delta_3,k}}-\delta_3&=&q^{2t}+q^{3}>0.
\end{eqnarray*}

(2.3): If $k=2$, we get that
\begin{eqnarray*}
y_{_{\delta_3,k}}&=&q^2\delta_3-\frac{q^{2}-q-2}{2}n\\
                 &=&\frac{1}{2}q^{2t+1}\!\!+\!\!\frac{1}{2}((q^{2t}-q^{2t-1}\!\!+\!\!\cdots-q^{3})\!\!-\!\!q^{2}+q+2)\!\!-\!\!q^{3}(q-1),\\
y_{_{\delta_3,k}}-\delta_3&=&q^{2t}-q^{4}+q^{3}+1>0,\\
n-y_{_{\delta_3,k}}-\delta_3&=&(q^{2t-1}-q^{2t-2}+\cdots+q^{3})+q(q^{2}+1)(q-1)>0.
\end{eqnarray*}

(2.4):  If $t\geq 3$ and $k=3,5,\cdots,2t-3$,  it then follows  that
\begin{eqnarray*}
y_{_{\delta_3,k}}
&=&q^k\delta_3-\frac{1}{2}(q^{k}-q^{k-1}-\frac{q(q^{k-2}+1)}{q+1})n\\
                 &=&\frac{1}{2}(q^{2t+1}-q^{2t}\cdots-q^{2}+q)-(q^{k+2}-q^{k+1}+q^{k}-q^{k-1}),\\
y_{_{\delta_3,k}}-\delta_3
&=&\!(q^{2t-1}\!-\!q^{2t-2}\!+\!q^{2t-3}\cdots+q)\!-\!(q^{k+2}\!-\!q^{k+1}\!+\!q^{k}\!-\!q^{k-1})\!\!+\!\!q(q-1)\\
&\geq&\!(q^{2t-1}\!-\!q^{2t-2}\!\cdots+q)\!\!-\!\!(q^{(2t-3)+\!2}\!\!-\!\!q^{(2t-3)+1}\!\!+\!\!q^{2t-3}\!\!-\!\!q^{(2t-3)\!-\!1})\!\!+\!\!q(q-1)\\
&=&\!\!(q^{2t-5}-q^{2t-6}\cdots+q)+q(q-1)>0,\\
%\end{eqnarray*}
%\begin{eqnarray*}
n-y_{_{\delta_3,k}}-\delta_3
&=&q^{2t}+(q^{k+2}-q^{k+1}+q^{k}-q^{k-1})+1+q(q-1)\\
&\geq&q^{2t}+(q^{3+2}-q^{3+1}+q^{3}-q^{3-1})+1+q(q-1)\\
&=&q^{2t}+q^{5}-q^{4}+q^{3}-q+1>0.
\end{eqnarray*}
(2.5):  If $t\geq 3$ and $k=4,6,\cdots,2t-2$,  then one can
similarly infer that
\begin{eqnarray*}
y_{_{\delta_3,k}}
&=&2^k\delta_3-\frac{1}{2}(q^{k}-q^{k-1}-q^{k-2}+\frac{q(q^{k-3}+1)}{q+1})n+n,\\
  &=&\frac{1}{2}q^{2t+1}\!+\!\frac{1}{2}(q^{2t}\!-\!q^{2t-1}\!+\!\cdots+q^{2}-q)-(q^{k+2}-q^{k+1}+q^{k}-q^{k-1})\!+\!1,
  \end{eqnarray*}
\begin{eqnarray*}
y_{_{\delta_3,k}}-\delta_3
&=&q^{2t}-(q^{k+2}-q^{k+1}+q^{k}-q^{k-1})+1+q(q-1)\\
&\geq&q^{2t}-(q^{(2t-2)+2}-q^{(2t-2)+1}+q^{2t-2}-q^{(2t-2)-1})+1+q(q-1)\\
&=&(q^{2t-3}+1)(q^{2}-q+1)>0,\\
n-y_{_{\delta_3,k}}-\delta_3
&=&(q^{2t-1}-q^{2t-2}\cdots-q^{2}+q)+(q^{k+2}-q^{k+1}+q^{k}-q^{k-1})+q(q-1)\\
&\geq&(q^{2t-1}-q^{2t-2}\cdots-q^{2}+q)+(q^{4+2}-q^{4+1}+q^{4}-q^{4-1})+q(q-1)\\
&=&(q^{2t-1}-q^{2t-2}\cdots+q^{5})+q^{6}-q^{5}>0.
\end{eqnarray*}
(2.6):  If $k=2t-1$,  we have then
\begin{eqnarray*}
y_{_{\delta_3,k}}&=&q^k\delta_3-\frac{1}{2}(q^{k}-q^{k-1}-\frac{q(q^{k-2}+1)}{q+1})n+n\\
                 &=&\frac{1}{2}(q^{2t+1}+q^{2t}-q^{2t-1}+q^{2t-2}+(q^{2t-3}\cdots-q^{2}+q))\!\!+\!\!1,\\
y_{_{\delta_3,k}}-\delta_3&=&q^{2t}+(q^{2t-3}-q^{2t-2}\cdots+q)+q^{2}-q+1>0,\\
n-y_{_{\delta_3,k}}-\delta_3&=&q^{2t-1}-q^{2t-2}+q^2-q>0.
\end{eqnarray*}
(2.7):  If $k=2t$, it is not difficult to obtain that
\begin{eqnarray*}
y_{_{\delta_3,k}}&=&q^k\delta_3-\frac{1}{2}(q^{k}-q^{k-1}-q^{k-2}+\frac{q(q^{k-3}+1)}{q+1})n+qn\\
                 &=&\frac{1}{2}(q^{2t+1}-q^{2t}+q^{2t-1}+(q^{2t-2}-q^{2t-3}+\cdots+q^{2})+q),\\
y_{_{\delta_3,k}}-\delta_3&=&q^{2t-1}+q^{2}>0,\\
n-y_{_{\delta_3,k}}-\delta_3&=&q^{2t}-(q^{2t-2}-q^{2t-3}\cdots+q^{2})+q(q-1)+1>0.
\end{eqnarray*}

To sum up, one shall know that $\delta_3$ is a coset leader.

In the similar way to the proofs of Steps 1 and 2, we can also
derive that both  $\delta_4$ and $\delta_5$ are coset leaders.
Therefore, the detailed proofs are omitted here.
\end{proof}
\subsection{The proof of Lemma \ref{lem5.3}}\label{plem5.3}

\begin{proof}
It shall be verified by the mathematical induction here.

If $r=2$, then $S^2=(1,-1)$, it is trivial that $F^2_{(1)}=(1,1)$
and $H^2_{(1)}=(1,-1)$. This obviously implies that  $F^2_{(1)}\geq
S_2$ and $H^2_{(1)}\geq S_2$.

%For clarity, we proceed to analyze the case for  $r=3$. If $r=3$,
%note that $S^3=(1,-1,-1,1)$. Then by definition we have
%
%$F^3_{(1)}=(1,1,-1,1)$ and $H^3_{(1)}=(1,1,-1,-1)$,
%
%  $F^3_{(2)}=(1,-1,1,-1)$ and $H^3_{(2)}=(1,-1,-1,1)$,
%
%  $F^3_{(3)}=(1,-1,1,1)$ and $H^3_{(3)}=(1,1,-1,-1)$.
%
%It is easy to get that $F^3_{(k)}\geq S_3$ and $H^3_{(k)}\geq S_3$
%for any $1\leq k\leq 2^{3-1}-1=3$.

Now  assume that $F^{r}_{(k)}\geq S^{r}$ and $H^{r}_{(k)}\geq S^{r}$
for $r\geq 3$ and $1\leq k\leq 2^{r-1}-1$. The remainder of the
proof is to verify $F^{r+1}_{(k)}\geq S^{r+1}$ and
$H^{r+1}_{(k)}\geq S^{r+1}$ for $1\leq k\leq 2^{r}-1$. By
definition, if
$S^r=(s^{r}_{_{2^{r-1}-1}},\cdots,s^{r}_{_{1}},s^{r}_{_{0}})$, then
 $$S^{r+1}=(S^{r},-S^{r})
=(s^{r}_{_{2^{r-1}-1}},\cdots,s^{r}_{_{1}},s^{r}_{_{0}},-s^{r}_{_{2^{r-1}-1}},\cdots,-s^{r}_{_{1}},-s^{r}_{_{0}}).$$
 We split  into the following cases.

 {\bf Case 1}: $1\leq k\leq 2^{r-1}-1$.

{\bf Subcase 1.1}: If
$s^{r+1}_{_{2^{r}-1-k}}=s^{r}_{_{2^{r-1}-1-k}}=1$,   we obtain that

$F^{r+1}_{(k)}=(s^{r}_{_{2^{r}-1-k}},\cdots,s^{r}_{_{1}},s^{r}_{_{0}},-s^{r}_{_{2^{r-1}-1}},\cdots,-s^{r}_{_{1}},-s^{r}_{_{0}},-s^{r}_{_{2^{r-1}-1}},-s^{r}_{_{2^{r-1}-2}},\cdots,-s^{r}_{_{2^{r-1}-k}})$

and
$H^{r+1}_{(k)}=(s^{r}_{_{2^{r}-1-k}},\cdots,s^{r}_{_{1}},s^{r}_{_{0}},-s^{r}_{_{2^{r-1}-1}},\cdots,-s^{r}_{_{1}},-s^{r}_{_{0}},s^{r}_{_{2^{r-1}-1}},s^{r}_{_{2^{r-1}-2}},\cdots,s^{r}_{_{2^{r-1}-k}})$.

Notice that $S^{r+1}=(S^{r},-S^{r})
=(s^{r}_{_{2^{r-1}-1}},\cdots,s^{r}_{_{1}},s^{r}_{_{0}},-s^{r}_{_{2^{r-1}-1}},\cdots,-s^{r}_{_{1}},-s^{r}_{_{0}}).$

According to the assumption, one has known that

$F^{r}_{(k)}=(s^{r}_{_{2^{r-1}-1-k}},\cdots,s^{r}_{_{1}},s^{r}_{_{0}},-s^{r}_{_{2^{r-1}-1}},-s^{r}_{_{2^{r-1}-2}},\cdots,-s^{r}_{_{2^{r-1}-k}})\geq
S^r$.

Clearly,  it follows that $F^{r+1}_{(k)}\geq S^{r+1}$ and
$H^{r+1}_{(k)}\geq S^{r+1}$.

{\bf Subcase 1.2}: If
$s^{r+1}_{_{2^{r}-1-k}}=s^{r}_{_{2^{r-1}-1-k}}=-1$, we have

$F^{r+1}_{(k)}=(-s^{r}_{_{2^{r}-1-k}},\cdots,-s^{r}_{_{1}},-s^{r}_{_{0}},s^{r}_{_{2^{r-1}-1}},\cdots,s^{r}_{_{1}},s^{r}_{_{0}},s^{r}_{_{2^{r-1}-1}},s^{r}_{_{2^{r-1}-2}},\cdots,s^{r}_{_{2^{r-1}-k}})$
and

$H^{r+1}_{(k)}=(-s^{r}_{_{2^{r}-1-k}},\cdots,-s^{r}_{_{1}},-s^{r}_{_{0}},s^{r}_{_{2^{r-1}-1}},\cdots,s^{r}_{_{1}},s^{r}_{_{0}},-s^{r}_{_{2^{r-1}-1}},-s^{r}_{_{2^{r-1}-2}},\cdots,-s^{r}_{_{2^{r-1}-k}}).$

Notice that $S^{r+1}=(S^{r},-S^{r})
=(s^{r}_{_{2^{r-1}-1}},\cdots,s^{r}_{_{1}},s^{r}_{_{0}},-s^{r}_{_{2^{r-1}-1}},\cdots,-s^{r}_{_{1}},-s^{r}_{_{0}}).$

By   assumption, we have known that

$F^{r}_{(k)}=(-s^{r}_{_{2^{r-1}-1-k}},\cdots,-s^{r}_{_{1}},-s^{r}_{_{0}},s^{r}_{_{2^{r-1}-1}},s^{r}_{_{2^{r-1}-2}},\cdots,s^{r}_{_{2^{r-1}-k}})\geq
S^r$.

Consequently, it is easy to know $F^{r+1}_{(k)}\geq S^{r+1}$ and
$H^{r+1}_{(k)}\geq S^{r+1}$.

 {\bf Case 2}: $k=2^{r-1}$.

 Notice that $s^{r+1}_{_{2^{r}-1-k}}=-s^{r}_{_{2^{r-1}-1}}=-1$. Then we have

$F_{r+1}^{(k)}=(s^{r}_{_{2^{r-1}-1}},\cdots,s^{r}_{_{1}},s^{r}_{_{0}},s^{r}_{_{2^{r-1}-1}},\cdots,s^{r}_{_{1}},s^{r}_{_{0}})=(S^r,S^r)$
and

 $H_{r+1}^{(k)}=(s^{r}_{_{2^{r-1}-1}},\cdots,s^{r}_{_{1}},s^{r}_{_{0}},-s^{r}_{_{2^{r-1}-1}},\cdots,-s^{r}_{_{1}},-s^{r}_{_{0}})=(S^r,-S^r)=S_{r+1}$

It   obviously follows that $F^{r+1}_{(k)}\geq S^{r+1}$ and
$H^{r+1}_{(k)}\geq S^{r+1}$.

{\bf Case 3}: $2^{r-1}+1\leq k\leq 2^{r}-1$. Put $u=k- 2^{r-1}$.
Then $1\leq u \leq 2^{r-1}-1$.

{\bf Subcase 3.1} If
$s^{r+1}_{_{2^{r}-1-k}}=-s^{r}_{_{2^{r-1}-1-u}}=1$, we  get that

$F^{r+1}_{(k)}
=(-s^{r}_{_{2^{r}-1-u}},\cdots,-s^{r}_{_{1}},-s^{r}_{_{0}},-s^{r}_{_{2^{r-1}-1}},\cdots,-s^{r}_{_{1}},-s^{r}_{_{0}},s^{r}_{_{2^{r-1}-1}},s^{r}_{_{2^{r-1}-2}},\cdots,s^{r}_{_{2^{r-1}-u}})$

and
 $H^{r+1}_{(k)}
=(-s^{r}_{_{2^{r}-1-u}},\cdots,-s^{r}_{_{1}},-s^{r}_{_{0}},s^{r}_{_{2^{r-1}-1}},\cdots,s^{r}_{_{1}},s^{r}_{_{0}},-s^{r}_{_{2^{r-1}-1}},-s^{r}_{_{2^{r-1}-2}},\cdots,-s^{r}_{_{2^{r-1}-u}})$

According to the assumption, one has known

$H^{r}_{(k)} =
(-s^{r}_{_{2^{r-1}-1-u}},\cdots,-s^{r}_{_{1}},-s^{r}_{_{0}},-s^{r}_{_{2^{r-1}-1}},-s^{r}_{_{2^{r-1}-2}},\cdots,-s^{r}_{_{2^{r-1}-u}})\geq
S^r$ and

 $F^{r}_{(k)}
=(-s^{r}_{_{2^{r-1}-1-u}},\cdots,,s^{r}_{_{1}},-s^{r}_{_{0}},s^{r}_{_{2^{r-1}-1}},s^{r}_{_{2^{r-1}-2}},\cdots,s^{r}_{_{2^{r-1}-u}})\geq
S^r$.

Combining  $S^{r+1}=(S^{r},-S^{r})
=(s^{r}_{_{2^{r-1}-1}},\cdots,s^{r}_{_{1}},s^{r}_{_{0}},-s^{r}_{_{2^{r-1}-1}},\cdots,-s^{r}_{_{1}},-s^{r}_{_{0}}),$
one can easily derive  that $F^{r+1}_{(k)}\geq S^{r+1}$ and
$H^{r+1}_{(k)}\geq S^{r+1}$.

{\bf Subcase 3.2}  If
$s^{r+1}_{_{2^{r}-1-k}}=-s^{r}_{_{2^{r-1}-1-u}}=-1$, we have

$F^{r+1}_{(k)}
=(s^{r}_{_{2^{r}-1-u}},\cdots,s^{r}_{_{1}},s^{r}_{_{0}},s^{r}_{_{2^{r-1}-1}},\cdots,s^{r}_{_{1}},s^{r}_{_{0}},-s^{r}_{_{2^{r-1}-1}},-s^{r}_{_{2^{r-1}-2}},\cdots,-s^{r}_{_{2^{r-1}-u}})$
and

 $H^{r+1}_{(k)}
=(s^{r}_{_{2^{r}-1-u}},\cdots,s^{r}_{_{1}},s^{r}_{_{0}},-s^{r}_{_{2^{r-1}-1}},\cdots,-s^{r}_{_{1}},-s^{r}_{_{0}},s^{r}_{_{2^{r-1}-1}},s^{r}_{_{2^{r-1}-2}},\cdots,s^{r}_{_{2^{r-1}-u}})$

Similarly,  we shall also get that $F^{r+1}_{(k)}\geq S^{r+1}$ and
$H^{r+1}_{(k)}\geq S^{r+1}$.

To sum up, one  can conclude that $F^{r+1}_{(k)}\geq S^{r+1}$ and
$H^{r+1}_{(k)}\geq S^{r+1}$  for any $1\leq k\leq 2^{r-1}$. This
completes the proof.
\end{proof}
%% References


\bibliographystyle{model1a-num-names}
\end{document}